\documentclass[11pt,letter]{article}

\everymath{\color{MidnightBlack}}

\usepackage{titling}
\thanksmarkseries{arabic}
\newcommand*\samethanks[1][\value{footnote}]{\footnotemark[#1]}

\RequirePackage[T1]{fontenc}
\RequirePackage[utf8]{inputenc}
\usepackage[greek,russian,english]{babel}
\usepackage{alphabeta}

\usepackage[pdftex,
backref=true,
plainpages = false,
pdfpagelabels,
hyperfootnotes=true,
pdfpagemode=FullScreen,
bookmarks=true,
bookmarksopen = true,
bookmarksnumbered = true,
breaklinks = true,
hyperfigures,
pagebackref,
urlcolor = magenta,
urlcolor = MidnightBlack,
anchorcolor = green,
hyperindex = true,
colorlinks = true,
linkcolor = black!30!blue,
citecolor = black!30!green
]{hyperref}

\usepackage[svgnames]{xcolor}
\usepackage{bm,fullpage,aliascnt,amsthm,amsmath,amsfonts,rotate,datetime,amsmath,ifthen,eurosym,wrapfig,cite,url,subcaption,cite,amsfonts,amssymb,ifthen,color,wrapfig,rotate,lmodern,aliascnt,datetime,graphicx,algorithmic,algorithm,enumerate,enumitem,todonotes,fancybox,cleveref,bm,subcaption,tabularx,colortbl,xspace,graphicx,algorithmic,algorithm}
\usepackage{stmaryrd,lmodern}

\newcommand{\qedclaim}{$\diamond$}
\newcommand{\wtm}{\leq_{\sf wtp}}

\newcommand{\ewidth}{\mathbf{p}_{\lambda^{\sf e}}}
\newcommand{\lwidth}{\mathbf{p}_{\delta^{\sf ce}}}
\newcommand{\name}[1]{\textsc{#1}}

\newcommand{\pbdefn}[3]{
\begin{tabbing}
\name{#1}\\
\emph{Input:} \hspace{1.2cm} \= \parbox[t]{14cm}{#2} \\
\emph{Question:}             \> \parbox[t]{14cm}{#3} \\
\end{tabbing}
}

\newcommand{\edeg}{{\sf edeg}}
\newcommand{\vdeg}{{\sf vdeg}}
\newcommand{\obs}{{\sf obs}}

\newcommand{\tw}{{\bf tw}}
\newcommand{\pw}{{\bf pw}}
\newcommand{\cw}{{\bf cw}}
\newcommand{\T}{{\sf T}}

\newcommand{\p}{{\bf p}}

\newcommand{\N}{\Bbb{N}}

\newcommand{\hh}{\end{document}}
\renewcommand{\O}{{\cal O}}

\newcommand{\remove}[1]{}

\newcommand{\NP}{{\sf NP}\xspace}
\newcommand{\FPT}{{\sf FPT}\xspace}
\newcommand{\XP}{{\sf XP}\xspace}
\newcommand{\W}{W\xspace}

\renewcommand{\L}{{\sf L}}

\newcommand{\tcw}{{\bf  tcw}}
\newcommand{\etw}{{\bf  etw}}

\RequirePackage{stmaryrd}
\usepackage{textcomp}
\DeclareUnicodeCharacter{2286}{\subseteq}
\DeclareUnicodeCharacter{2192}{\ifmmode\to\else\textrightarrow\fi}
\DeclareUnicodeCharacter{2203}{\ensuremath\exists}
\DeclareUnicodeCharacter{183}{\cdot}
\DeclareUnicodeCharacter{2200}{\forall}
\DeclareUnicodeCharacter{2264}{\leq}
\DeclareUnicodeCharacter{2265}{\geq}
\DeclareUnicodeCharacter{8614}{\mathbin{\mapsto}}
\DeclareUnicodeCharacter{8656}{{\sf L}eftarrow}
\DeclareUnicodeCharacter{8657}{\Uparrow}
\DeclareUnicodeCharacter{8658}{\Rightarrow}
\DeclareUnicodeCharacter{8659}{\Downarrow}
\DeclareUnicodeCharacter{8669}{\rightsquigarrow}
\newcommand{\eqdef}{\stackrel{{\scriptsize\rm def}}{=}}
\DeclareUnicodeCharacter{8797}{\eqdef}
\DeclareUnicodeCharacter{8870}{\vdash}
\DeclareUnicodeCharacter{8873}{\Vdash}
\DeclareUnicodeCharacter{22A7}{\models}
\DeclareUnicodeCharacter{9121}{\lceil}
\DeclareUnicodeCharacter{9123}{\lfloor}
\DeclareUnicodeCharacter{9124}{\rceil}
\DeclareUnicodeCharacter{2208}{\in}
\DeclareUnicodeCharacter{9126}{\rfloor}
\DeclareUnicodeCharacter{9655}{\triangleright}
\DeclareUnicodeCharacter{9665}{\triangleleft}
\DeclareUnicodeCharacter{9671}{\diamond}
\DeclareUnicodeCharacter{9675}{\circ}
\DeclareUnicodeCharacter{10178}{\bot}
\DeclareUnicodeCharacter{10214}{} 
\DeclareUnicodeCharacter{10215}{} 
\DeclareUnicodeCharacter{10229}{\longleftarrow}
\DeclareUnicodeCharacter{10230}{\longrightarrow}
\DeclareUnicodeCharacter{10231}{\longleftrightarrow}
\DeclareUnicodeCharacter{10232}{{\sf L}ongleftarrow}
\DeclareUnicodeCharacter{10233}{{\sf L}ongrightarrow}
\DeclareUnicodeCharacter{10234}{{\sf L}ongleftrightarrow}
\DeclareUnicodeCharacter{10236}{\longmapsto}
\DeclareUnicodeCharacter{10238}{{\sf L}ongmapsto} 
\DeclareUnicodeCharacter{10503}{\Mapsto}    
\DeclareUnicodeCharacter{10971}{\mathrel{\not\hspace{-0.2em}\cap}}
\DeclareUnicodeCharacter{65294}{\ldotp}
\DeclareUnicodeCharacter{65372}{\mid}

\usepackage{color}
\definecolor{MidnightBlack}{rgb}{0.1,0.1,.34}
\definecolor{MidnightBlue}{rgb}{0.1,0.1,0.44}
\definecolor{Black}{rgb}{0,0, 0}
\definecolor{Blue}{rgb}{0, 0 ,1}
\definecolor{Red}{rgb}{1, 0 ,0}
\definecolor{White}{rgb}{1, 1, 1}
\definecolor{Grey}{rgb}{.6, .6, .6}
\definecolor{Mygreen}{rgb}{.0, .7, .0}
\definecolor{Yellow}{rgb}{.55,.55,0}
\definecolor{Mustard}{rgb}{1.0, 0.86, 0.35}
\definecolor{applegreen}{rgb}{0.55, 0.71, 0.0}
\definecolor{darkturquoise}{rgb}{0.0, 0.81, 0.82}
\definecolor{celestialblue}{rgb}{0.29, 0.59, 0.82}
\definecolor{green_yellow}{rgb}{0.68, 1.0, 0.18}
\definecolor{crimsonglory}{rgb}{0.75, 0.0, 0.2}
\definecolor{darkmagenta}{rgb}{0.30, 0.0, 0.30}
\definecolor{internationalorange}{rgb}{1.0, 0.31, 0.0}
\definecolor{darkorange}{rgb}{1.0, 0.55, 0.0}
\definecolor{ao}{rgb}{0.0, 0.5, 0.0}
\definecolor{awesome}{rgb}{1.0, 0.13, 0.32}

\newcommand{\black}[1]{{\color{Black}#1}}

\newcommand{\blue}[1]{{\color{Blue}#1}}
\newcommand{\red}[1]{{\color{Red}#1}}
\newcommand{\green}[1]{{\color{Mygreen}#1}}

\usepackage{tikz}
\usetikzlibrary{patterns}
\usetikzlibrary{calc,3d}
\usetikzlibrary{intersections,decorations.pathmorphing,shapes,decorations.pathreplacing,fit}
\usetikzlibrary{shapes.geometric,hobby,calc}
\usetikzlibrary{decorations}
\usetikzlibrary{decorations.pathmorphing}
\usetikzlibrary{decorations.text}
\usetikzlibrary{shapes.misc}
\usetikzlibrary{decorations,shapes,snakes}

\newcounter{func}

\newcommand{\funref}[1]{\hyperref[#1]{f_{\ref*{#1}}}} 
\newcounter{con}

\newcommand{\conref}[1]{\hyperref[#1]{c_{\ref*{#1}}}} 

\newcommand{\mynewtheorem}[2]{
	\newaliascnt{#1}{dummy}
	\newtheorem{#1}[#1]{#2}
	\aliascntresetthe{#1}
}

\theoremstyle{plain}
\mynewtheorem{theorem}{Theorem}
\mynewtheorem{proposition}{Proposition}
\mynewtheorem{corollary}{Corollary}
\mynewtheorem{lemma}{Lemma}
\theoremstyle{definition}
\mynewtheorem{definition}{Definition}
\theoremstyle{remark}
\mynewtheorem{sublemma}{Sublemma}
\mynewtheorem{claim}{Claim}
\mynewtheorem{remark}{Remark}
\mynewtheorem{fact}{Fact}
\mynewtheorem{conjecture}{Conjecture}
\mynewtheorem{question}{Question}
\mynewtheorem{observation}{Observation}
\mynewtheorem{problem}{Problem}
\mynewtheorem{note}{Note}

\addto\extrasenglish{

}

\makeatletter
\newtheoremstyle{caja1}
  {\topsep}
  {\topsep}
  {\itshape}
  {}
  {}
  {}
  {.5em}
  {\blue{\fbox{\black{\thmname{#1}~\thmnumber{#2}\@ifempty{#3}{.}{}\thmnote{ (#3).}}}}}
\makeatother
\theoremstyle{caja1}

\makeatletter
\newtheoremstyle{caja2}
  {\topsep}
  {\topsep}
  {\itshape}
  {}
  {}
  {}
  {.5em}
  {\green{\fbox{\black{\thmname{#1}~\thmnumber{#2}\@ifempty{#3}{.}{}\thmnote{ (#3).}}}}}
\makeatother
\theoremstyle{caja2}

\makeatletter
\newtheoremstyle{caja3}
  {\topsep}
  {\topsep}
  {\itshape}
  {}
  {}
  {}
  {.5em}
  {\red{\fbox{\black{\thmname{#1}~\thmnumber{#2}\@ifempty{#3}{.}{}\thmnote{ (#3).}}}}}
\makeatother
\theoremstyle{caja3}

\newtheorem{innercustomgeneric}{\customgenericname}
\providecommand{\customgenericname}{}

\begin{document}

\title{\bf Edge-treewidth: Algorithmic and\\  combinatorial  properties\thanks{The second and the last author were supported  by   the ANR projects DEMOGRAPH (ANR-16-CE40-0028), ESIGMA (ANR-17-CE23-0010), and the French-German Collaboration ANR/DFG Project UTMA (ANR-20-CE92-0027). Emails: \texttt{sedthilk@thilikos.info}}}
\author{Loïc Magne\thanks{École normale supérieure,  de l'Université Paris-Saclay, France.} \and Christophe Paul\thanks{LIRMM, Univ Montpellier, CNRS, Montpellier, France.} \and Abhijat Sharma\thanks{IDA, AIICS, Linköping University, Sweden.} \and Dimitrios M. Thilikos\samethanks[3]}

\maketitle

\begin{abstract}
\noindent We introduce the graph theoretical parameter of edge treewidth. This parameter occurs in a natural way as the tree-like analogue of cutwidth or, alternatively, as an edge-analogue of treewidth. 
We study  the combinatorial properties of edge-treewidth. We first observe that edge-treewidth does not enjoy any closeness properties under the known partial ordering relations on graphs. We introduce a variant of the topological minor relation, namely, the {\em weak topological minor} relation and we prove that edge-treewidth is closed under weak topological minors. Based on this new relation we are able to provide universal 
obstructions for edge-treewidth. The proofs are based on the fact that  edge-treewidth of a graph is parametetrically equivalent with the maximum over the treewidth and the maximum degree of the blocks of the graph. 
We also prove that deciding whether the edge-treewidth  of a graph is at most $k$ is an {\sf NP}-complete problem.
\end{abstract}
\bigskip

\tableofcontents

\newpage
\section{Introduction}

A vibrant area of research in graph algorithms is dedicated to the study of structural  graph parameters and their algorithmic applications. Perhaps the most prominent graph parameter is treewidth (hereafter denoted by ${\tw}$).  The importance of  treewidth resides on both its combinatorial and algorithmic applications. From the algorithmic point of view, treewidth
has readily became important due to Courcelle's theorem~\cite{Courcelle90them}, asserting that 
every problem definable by some sentence $φ$ in Monadic Second Order Logic (MSOL) can be solved in time $\O_{\tw(G),|φ|}(|G|)$.\footnote{
Let ${\bf t}=(x_{1},\ldots,x_{l})∈ \N^l$ and $\chi,ψ: \N
\rightarrow \N.$
We say that $\chi(n)=\O_{{\bf t}}(ψ(n))$ if there exists a computable
function $φ:\N^{l} \rightarrow \N$
such that  $\chi(n)=\O(φ({\bf t})\cdot ψ(n)).$}
Using the terminology of parameterized complexity, this means that every MSOL-definable problem
admits a fixed parameter algorithm (in short, an {\sf FPT}-algorithm) when parameterized by treewidth.

Interestingly, there are  several problems where Courcelle's theorem does not apply. 
As mentioned in~\cite{GanianKS15algor}, problems  such as {\sc Capacitated Vertex Cover}, {\sc Capacitated Dominating Set},  {\sc List Coloring}, and {\sc Boolean CSP} 
are ${\sf W}[1]$-hard, when parameterized by $\tw$, which mens that an {\sf FTP}-algorithm may not be expected for them. The emerging question for such problems is whether they may admit an {\sf FTP}-algorithm  when parameterized by some alternative structural  parameter.  As many problems escaping the expressibility 
power of MSOL are defined using certificates that are edge sets, the investigation has been oriented to edge-analogues of treewidth.  The first candidate 
for this was the parameter of {\sl tree-cut width}, denoted by {\bf twc}, defined by Wollan in~\cite{Wollan15thest}.
Indeed tree-cut width enjoys combinatorial properties that parallel those of treewidth and, most importantly, 
it yelds {\sf FPT}-algorithms for several such problems including the above mentioned ones~\cite{GanianKS15algor}. Interestingly, this landscape appears to be more complicated when it comes 
to the {\sc Edge Disjoint Paths} problem. Ganian and Ordyniak, in~\cite{GanianO21ThePo}, proved 
that this problem is ${\sf W}[1]$-hard when parameterized by {\bf twc}. This means that some other, alternative to {\bf twc}, parameter should be considered for this problem. This paper was motivated by this question.
We give the definition of a different parameter, called  {\sl edge-treewidth} and denoted by \etw. As we see, edge-treewidtht is parametrically incomparable to tree-cut width and can be seen in a natural way as an ``edge-analogue'' of 
treewidth or as a ``tree-analogue'' of cutwidth, based on their layout definitions.
Moreover, it seems that it is the right choice as a parameter for the {\sc Edge Disjoint Paths} problem: this problem
admits an {\sf FPT}-algorithm when parameterized by \etw.

%

\paragraph{Some definitions on graphs and on graph parameters.} 
Before we proceed with  the definition of edge-treewidth and its relation to other parameters, we need some definitions.

All graphs that  we consider are finite, loop-less, and may have multiple edges. Given a graph $G$, we let $V(G)$ and $E(G)$ respectively denote its vertex set and edge set. As we permit multi-edges, $E(G)$ is a multiset and the {\em multiplicity} of each edge is the number of times that it appears in $E(G)$.
We use  $|G|=|V(G)|$ in order to denote the {\em size} of $G$.
 For a subset of vertices $S\subseteq V(G)$, we denote by  $G[S]$ the subgraph of $G$ induced by $S$. We also define $E_{G}(S)$ as the set of edges with one vertex in $S$ and one vertex not in $S$
and $N_{G}(S)$ as the set of all vertices not in $S$ that are adjacent to a vertex in $S$.
The {\em vertex-degree} (resp. {\em edge-degree}) of a vertex $v$ is defined as ${\sf vdeg}_{G}(v)=|N_{G}(\{v\})|$ (resp. ${\sf edeg}_{G}(v)=|E_{G}(\{v\})|$). We also set 
$\Delta_{\sf v}(v)=\max\{{\sf vdeg}_{G}(v)\mid v\in V(G)\}$ and 
$\Delta_{\sf e}(v)=\max\{{\sf edeg}_{G}(v)\mid v\in V(G)\}$.
Given a graph $G$, a set of vertices $S$ and a vertex $v\in S$, we define 
$C_{G}(S,v)$ as the vertex set of the connected component of $G[S]$ containing the vertex $v$.

\paragraph{Graph parameters.}
A {\em graph parameter} is any function mapping graphs to non-negative integers.
Let $\p_{1},\p_{2}$ be two graph parameters. We write $\p_{1}=\p_{2}$ if, for every graph $G$, $\p_{1}(G)=\p_{2}(G)$.
We say that  $\p_{1}$ is {\em upper bounded} by  $\p_{2}$, denoted by $\p_{1}\lesssim\p_{2}$,  if there is a 
function $f:\Bbb{N}\to\Bbb{N}$ such that for every 
graph $G$, $\p_{1}(G)\leq f(\p_{2}(G))$.  We also say that $\p_{1}$ and $\p_{2}$ are {\em parametrically equivalent}, denoted by $\p_{1}\sim \p_{2}$, if  $\p_{2}\lesssim \p_{1}$
and $\p_{1}\lesssim \p_{2}$. Also, we say that  $\p_{1}$ and $\p_{2}$ are {\em parametrically incomparable}, denoted by $\p_{1}\not\sim \p_{2}$, if  neither $\p_{2}\lesssim \p_{1}$
nor $\p_{1}\lesssim \p_{2}$. 

\vspace{-0.2cm}
\paragraph{Layouts.}
 A \emph{layout}  of an $n$-vertex graph $G$ is a linear ordering ${\sf L}=\langle x_{1},\ldots,x_{n}\rangle$ of its vertices.  
We  define $\mathcal{L}(G)$ as the set of layouts of $G$. For $i\in\{1,\dots,n\}$, we let ${{\sf L}(i)}$ denote the $i$-th vertex of ${\sf L}$. If clear from the context, ${\sf L}(i)$ may be denoted $x_i$. For two vertices $x$ and $y$, we write $x\prec_{{\sf L}} y$ if $x$ occurs before $y$ in ${\sf L}$.
Given ${\sf L}=\langle x_{1},\ldots,x_{n}\rangle\in {\cal L}(G)$ and  $i\in\{1,\ldots,n\}$, we define $S_{i}=\{x_i,\ldots,x_n\}$. Finally, if $S\subseteq V(G)$, we denote by ${{\sf L}[S]}$ the layout of $G[S]$ such that, for every $x,y\in S$, $x\prec_{{\sf L}[S]} y$ if and only if $x\prec_{{\sf L}} y$.

Given an $n$-vertex graph $G$ and a layout ${\sf L}=\langle x_{1},\ldots,x_{n}\rangle\in  \mathcal{L}(G)$, we call {\em $(G,{\sf L})$-cost function} every  
function $\delta_{G,{\sf L}}:[n]\to \Bbb{N}$   that assigns a non-negative integer to each position of the layout ${\sf L}$. From now on, for simplicity, we will use $\delta$ instead of the heavier notation $\delta_{G,{\sf L}}$.
Given a $(G,{\sf L})$-cost  function
$\delta$, we set ${\bf p}_{\delta}(G,{\sf L})=\max\{\delta(i)\mid i\in\{1,\ldots,n\}\}$
and we define the following general graph parameter:
$$ {\bf p}_{\delta}(G)=\min\{{\bf p}_{\delta}(G,{\sf L}) \mid {\sf L}\in \mathcal{L}(G)\}.$$
Four cost functions on layouts of graphs that are of particular interest are presented in~\autoref{costfusncexample}.

\vspace{-0.3cm}
\begin{figure}[h]
\begin{align*}
\delta^{\sf v}(i)  =  |N_{G}(S_{i})|         & & \delta^{\sf vc}(i) =  |N_{G}(C_{G}(S_i,x_{i}))|       \\
\delta^{\sf e}(i)  =  |E_{G}(S_{i})|         & & \red{\fbox{$δ^{\sf ec}(i)  =  |E_{G}(C_{G}(S_i,x_{i}))|$}}
\end{align*}
\caption{
The four cost functions $\delta^{\sf v}, \delta^{\sf vc}, \delta^{\sf e}$, and \red{\fbox{$δ^{\sf ec}$}}. It follows that ${\bf p}_{\delta^{\sf v}}={\bf pw}$, ${\bf p}_{\delta^{\sf vc}}={\bf tw}$, ${\bf p}_{\delta^{\sf e}}={\bf cw}$, and ${\bf p}_{δ^{\sf ec}}={\bf etw}$.}
\label{costfusncexample}
\end{figure}

\vspace{-0.2cm}
Notice that the cost functions in the left column are defined using the set $S_{i}$,  
while those in the right column are defined by the set $C_{G}(S_i,x_{i})$, that instead 
of $S_{i}$, takes into account the connected component of $G[S_{i}]$ that contains $x_{i}$.
On the other side, the first line defines the cost as the number of \textsl{vertices} adjacent to this  set, while the second line defines it by the number of \textsl{edges} incident to this set.\
See \autoref{costfuncexample} for an example of a layout of a graph and the corresponding cost functions.
\medskip

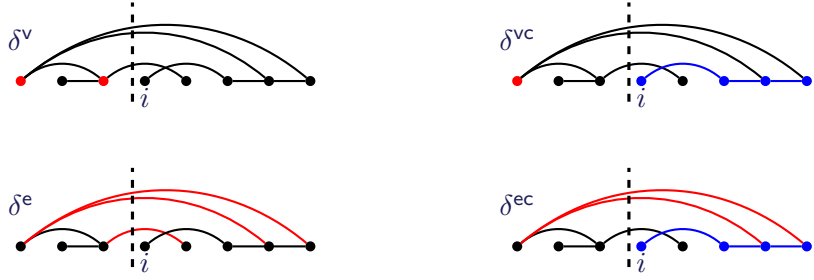
\begin{figure}[h]
\centering
\begin{center}
\begin{tikzpicture}[thick,scale=0.55]
\tikzstyle{sommet}=[circle, draw, fill=black, inner sep=0pt, minimum width=3pt]

\begin{scope}[xshift=0cm,yshift=0cm]
\draw[red] node[sommet,fill=red] (1) at (0,0){};
\draw node[sommet] (2) at (1,0){};
\draw[red] node[sommet,fill=red] (3) at (2,0){};
\draw node[sommet] (4) at (3,0){};
\draw node[sommet] (5) at (4,0){};
\draw node[sommet] (6) at (5,0){};
\draw node[sommet] (7) at (6,0){};
\draw node[sommet] (8) at (7,0){};

\draw[-] (2) to (3) ;
\draw[-] (6) to (7) ;
\draw[-] (7) to (8) ;
\draw[-] (1) to[bend left=40] (3) ;
\draw[-] (3) to[bend left=40] (5) ;
\draw[-] (4) to[bend left=40] (6) ;
\draw[-] (1) to[bend left=40] (7) ;
\draw[-] (1) to[bend left=40] (8) ;
\draw[-,dashed,very thick] (2.7,-0.5) to (2.7,1.9);

\node[] (a) at (3,-0.4) {$i$};
\node[] (a) at (0,1) {$\delta^{\sf v}$};
\end{scope}

\begin{scope}[xshift=12cm,yshift=0cm]
\draw[red] node[sommet,fill=red] (1) at (0,0){};
\draw node[sommet] (2) at (1,0){};
\draw[] node[sommet] (3) at (2,0){};
\draw[blue] node[sommet,fill=blue] (4) at (3,0){};
\draw node[sommet] (5) at (4,0){};
\draw[blue] node[sommet,fill=blue] (6) at (5,0){};
\draw[blue] node[sommet,fill=blue] (7) at (6,0){};
\draw[blue] node[sommet,fill=blue] (8) at (7,0){};

\draw[-] (2) to (3) ;
\draw[-,blue] (6) to (7) ;
\draw[-,blue] (7) to (8) ;
\draw[-] (1) to[bend left=40] (3) ;
\draw[-] (3) to[bend left=40] (5) ;
\draw[-,blue] (4) to[bend left=40] (6) ;
\draw[-] (1) to[bend left=40] (7) ;
\draw[-] (1) to[bend left=40] (8) ;
\draw[-,dashed,very thick] (2.7,-0.5) to (2.7,1.9);

\node[] (a) at (3,-0.4) {$i$};
\node[] (a) at (0,1) {$\delta^{\sf vc}$};
\end{scope}

\begin{scope}[yshift=-4cm]
\draw[] node[sommet] (1) at (0,0){};
\draw node[sommet] (2) at (1,0){};
\draw[] node[sommet] (3) at (2,0){};
\draw[] node[sommet] (4) at (3,0){};
\draw node[sommet] (5) at (4,0){};
\draw[] node[sommet] (6) at (5,0){};
\draw[] node[sommet] (7) at (6,0){};
\draw[] node[sommet] (8) at (7,0){};

\draw[-] (2) to (3) ;
\draw[-] (6) to (7) ;
\draw[-] (7) to (8) ;
\draw[-] (1) to[bend left=40] (3) ;
\draw[-,red] (3) to[bend left=40] (5) ;
\draw[-] (4) to[bend left=40] (6) ;
\draw[-,red] (1) to[bend left=40] (7) ;
\draw[-,red] (1) to[bend left=40] (8) ;
\draw[-,dashed,very thick] (2.7,-0.5) to (2.7,1.9);

\node[] (a) at (3,-0.4) {$i$};
\node[] (a) at (0,1) {$\delta^{\sf e}$};
\end{scope}

\begin{scope}[xshift=12cm,yshift=-4cm]
\draw[] node[sommet] (1) at (0,0){};
\draw node[sommet] (2) at (1,0){};
\draw[] node[sommet] (3) at (2,0){};
\draw[blue] node[sommet,fill=blue] (4) at (3,0){};
\draw node[sommet] (5) at (4,0){};
\draw[blue] node[sommet,fill=blue] (6) at (5,0){};
\draw[blue] node[sommet,fill=blue] (7) at (6,0){};
\draw[blue] node[sommet,fill=blue] (8) at (7,0){};

\draw[-] (2) to (3) ;
\draw[-,blue] (6) to (7) ;
\draw[-,blue] (7) to (8) ;
\draw[-] (1) to[bend left=40] (3) ;
\draw[-] (3) to[bend left=40] (5) ;
\draw[-,blue] (4) to[bend left=40] (6) ;
\draw[-,red] (1) to[bend left=40] (7) ;
\draw[-,red] (1) to[bend left=40] (8) ;
\draw[-,dashed,very thick] (2.7,-0.5) to (2.7,1.9);

\node[] (a) at (3,-0.4) {$i$};
\node[] (a) at (0,1) {$δ^{\sf ec}$};\end{scope}

\end{tikzpicture}
\end{center}
\caption{A layout of a 6 vertex graphs and the four cost functions $\delta^{v}$, $\delta^{\sf vc}$ (at the top), $\delta^{\sf e}$ and $δ^{\sf ec}$ at the bottom. In blue, we have the set $C_G(S_i,x_i)$. In red, we have the vertices of $N_{G}(S_{i})$, $N_{G}(C_{G}(S_i,x_{i}))$ and the edges of $E_{G}(S_{i})$, $E_{G}(C_{G}(S_i,x_{i}))$.
}
\label{costfuncexample}
\end{figure}

The \emph{treewidth} of a graph  measures the topological resemblance of a graph to the structure of a tree \cite{RoberstonS84GMIII,BerteleB73onno,Halin76sfun}. The \emph{pathwidth} of a graph, denoted by ${\bf pw}$, measures the topological resemblance of a graph to the structure of a path \cite{RobertsonS83GMI,EllisST94thev}. It holds that ${\bf p}_{\delta^{\sf v}}={\bf pw}$ \cite{Kinnersley92thev} and ${\bf p}_{\delta^{\sf vc}}={\tw}$ \cite{DendrisKT97fugi}. Finally, 
the parameter ${\bf p}_{\delta^{\sf e}}$ is known as the \emph{cutwidth} of a graph~\cite{Chung85onthe,HortonPB00onmi}, denoted ${\bf cw}$. From the layout characterization of pathwidth, we can consider the cutwidth as the  edge-analogue of pathwidth. Likewise, we may see ${\bf p}_{δ^{\sf ec}}$ as the edge-analogue of treewidth. We call ${\bf p}_{δ^{\sf ec}}$  \emph{edge-treewidth} and we denote it by ${\etw}$.
Interestingly, this parameter is, so far, completely non-investigated. The purpose of this note is to initiate  the study of edge-treewidth, as an ``edge-analogue'' of treewidth.

\medskip

As we already mentioned, an alternative edge-analogue of treewidth is {\em tree-cut width}.  The reason why tree-cut width can be seen as an edge-analogue of treewidth is that it is parametrically equivalent to the 
maximum size of a wall contained as an immersion in the same way that treewidth is parametrically equivalent to the maximum  
size of a wall contained as topological minor~\cite{Wollan15thest} (see \autoref{dadyboo}  for the formal definitions of these graph containment relations). 
Using the terminology that we introduce in \autoref{opertlsok}, the family of walls may serve as {\sl universal minor-obstruction} for 
treewidth as well as universal immersion-obstruction for tree-cut width.

All  four parameters ${\bf tw}, {\bf pw}, {\bf cw}, {\bf tcw}$
have been extensively studied both from the combinatorial and the algorithmic point of view \cite{GiannopoulouPRT17linea,Wollan15thest,Courcelle90themo,Bodlaender88dyna,Bodlaender96alin,Bodlaender98apa}. For all of them, the corresponding decision problem is  \NP{}-complete \cite{Wollan15thest,BodlaenderGHK95appr,ArnborgCP87comp,GareyJ79comp,BodlaenderT97treew}. Moreover, all of them   enjoy nice closeness properties under known 
partial ordering relations on graphs: treewidth and pathwidth are  minor-closed, while cutwidth and tree-cut  are  immersion-closed. These closeness relations imply that all of them
are fixed parameter tractable (in short, {\sf FPT}) when parameterized by their values. In other words, for every ${\bf p}\in\{{\bf tw}, {\bf pw}, {\bf cw}, {\bf tcw}\}$, there is an algorithm computing $\p(G)$   in  time $\O_{\p(G)}(|G|)$  \cite{BodlaenderK96effi,BodlaenderFT09deri,ThilikosSB05cutwI,ThilikosSB05acutwII,BodlaenderT98comp}. (For more on parameterized algorithms and complexity, see \cite{DowneyF99param,FlumG06param,CyganFK14param}.)

\vspace{-0.1cm}
\paragraph{Our results.}
Unfortunately, the above  combinatorial/algorithmic properties do not copy for $\etw$.
As we see in \autoref{dadyboo}, none of the above closeness properties holds for edge-treewidth.
In \autoref{dadyboo} we define a new partial ordering relation, namely the {\em weak topological minor}
relation 
and we show that edge-treewidth is closed under this relation. 
Let ${\cal G}_{k}=\{G\mid \etw(G)\leq k\}$, $k\in\Bbb{N}$.
We consider  the  {\sl weak-topological minor  obstruction set}
 of ${\cal G}_{k}$ that is defined as  the set $\obs({\cal G}_{k})$ containing all weak-topological minor minimal graphs that do not belong  in ${\cal G}_{k}$. In \autoref{uiower}, we identify the (finite) $\obs({\cal G}_{k})$ for $k\leq 2$ 
and {we show that, for $k=3$, $\obs({\cal G}_{k})$  is infinite. 
This comes in (negative) contrast to the fact that the corresponding obstructions sets for the graphs of treewidth/pathwidth/cutwidth at most $k$ is always finite because of the seminal results of Robertson and Semour in their Graph Minors series~\cite{RobertsonS04GMXX}. On the positive side, we prove that 
edge-treewidth admits a {\sl finite} ``long term'' obstruction characterization. As our main combinatorial  result (\autoref{paplriok}), 
we  give a set of  {\sl five} parameterized  
graph families that can act as {\sl universal obstructions} for edge-treewidth in the sense 
that edge-treewidth is parametrically equivalent to the maximum size of a graph in these families that 
is contained as weak-topological minor.

\medskip
The algorithmic motivation of our study of edge treewidth is the following problem.

\vspace{-0.2cm}
\pbdefn{Edge Disjoint Path problem (EDP)}{A graph $G$ and a set $(s_1,t_1),\dots,(s_k,t_k)$ of pairs of vertices, called terminals, of $G$.}{Does $G$ contains pairwise edge disjoint paths between every pair of terminals?}
\vspace{-3mm}

The vertex disjoint counterpart of {\sc EDP} is the {\sc Vertex Disjoint Path}  problem  ({\sc VDP}).
Both {\sc EDP} and {\sc VDP} are  known to be \NP{}-Complete~\cite{Karp75OnThe}. Both problems
are   {\sf FPT} when parameterized by the number $k$ of the terminals~\cite{RobertsonS95GMVIII}. Another relevant question is whether (and when)  {\sc VDP} and {\sc EDP}  remain \FPT{} when parameterized by certain graph parameters that are independent from the number of terminals.
In other words, the question is for which parameter $\p$ the above problems can be solved in time $|G|^{\O_{\p(G)}(1)}$ or, even better, in time $\O_{\p(G)}(|G|^{\O(1)})$.
In the former case, the corresponding parameterized problem is classified in the parameterized complexity class {\sf XP} and, as we already mentioned, in the later in the parameterized parameterized complexity class {\sf FPT}.

Naturally, the first parameter to consider is treewidth. It is known that {\sc VDP}, parameterized by  treewidth,
is in {\sf FPT} because of the  time  $\O_{\tw(G)}(|G|)$  algorithm of Scheffler in~\cite{Scheffler94apra}. However, this is not any more the case for {\sc EDP}
as it has been shown by Nishizeki, Vygen, and Zhou in~\cite{NishizekyVZ01TheEd} that the {\sc EDP} is \NP-hard even on graphs of treewidth two. Therefore, one may not expect that {\sc EDP},  parameterized by  treewidth,  is in \FPT{} or even in \XP{}. This induces the question on whether there is some alternative graph parameter ${\p}$ such that {\sc EDP} is {\sf FPT}, 
when parameterized by $\p$.  
The first parameter to be considered was tree-cut width, seen as an edge-analogue of treewidth. 
In this direction  Ganian and Ordyniak, in~\cite{GanianO21ThePo},  proved that {\sc EDP}, when parameterized by $\tcw,$ is in {\sf XP}, however, it is also $\W[1]$-hard, therefore we may not expect that a time $\O_{\tcw(G)}(|G|^{\O(1)})$ algorithm exists for this problem (i.e., an \FPT-algorithm). This means that another graph parameter should be considered so as to derive 
an {\sf FPT} algorithm for  {\sc EDP}. Also, as we next see,  \etw\ may play this role. Moreover, in the end of~\autoref{nioklop}, we show that edge-tree width is {parametrically  incomparable} to tree-cut width. \medskip

Given a graph $G$, a {\em block} of $G$ is either a bridge\footnote{An edge in a graph $G$ is a {\em bridge} of its removal increases the number of connected components.}  co of $G$ or a maximal 2-connected subgraph of $G$.
We denote by ${\sf bc}(G)$ as the set of all blocks of $G$.  In \autoref{nioklop}. we prove that 
  $\etw$  is parametrically equivalent to the graph parameter $\p,$ where
 \begin{eqnarray}
  \p(G)=\max\big\{\tw(B),\Delta_{\sf e}(B) \:|\: B \in {\sf bc}(G)\big\}\label{lkolpio}
\end{eqnarray}

It is not difficult to see that solving {\sc EDP} on a graph $G$ can be linearly reduced to the blocks of $G$.
Also by a simple reduction of   {\sc EDP}  to the  {\sc VDP},
 {\sc EDP} can be solved in time $\O_{\tw(G),\Delta_{\sf e}(G)}(|G|)$ (see e.g.,~\cite{GanianOR21OnStr}). These facts together imply a
 time  $\O_{\etw(G)}(|G|)$ algorithm for  {\sc EDP}. These two facts imply that {\sc EDP} is in \FPT\ when parameterized by  edge-treewidth. This indicates that edge-treewidth is the accurate edge analogue parameter to treewidth, when it comes to the study of {\sc EDP}  from the parameterized complexity point of view. As a byproduct of our results, it follows that an approximate estimation of \etw\ can be found in \FPT-time: there is an $\O_{k}(|G|)$  time  algorithm that, given a graph $G$ and an integer $k\in \Bbb{N}$,
outputs  either a report  that $\etw(G)>k$ or a layout ${\sf L}\in{\cal L}(G)$ certifying that $\etw(G)=\O(k^4)$.

Our paper is organized as follows. In~\autoref{dadyboo} we introduce the weak topological minor relation and 
we prove that the edge treewidth parameter is closed under this relation. In~\autoref{uiower} 
we identify the obstruction set (with respect to weak topological minors) 
for the graphs with edge treewidth at most two. In~\autoref{nioklop} we prove the parametric 
equivalence of edge treewidth with the parameter defined in~\eqref{lkolpio} and in~\autoref{opertlsok}
we use this equivalence in order to provide a universal obstruction set for edge treewidth, under the weak topological minor relation. In~\autoref{npcomoplo} we prove that checking whether $\etw(G)\leq k$
is an {\sf NP}-complete problem. 
\section{Characterizations of edge-treewidth}
\label{dadyboo}
We denote by $\Bbb{N}$ the set of non-negative integers.
Given two integers $p$ and
$q,$ the set $[p,q]$ refers to the set of every integer $r$ such that $p ≤ r ≤ q.$
For an integer $p≥ 1,$ we set $[p]=[1,p]$.

\paragraph{Tree-width and edge-treewidth by means of tree layouts}

A {\em rooted tree} is a pair $(T,r)$ where $T$ is a tree and $r$ a distinguished node called the {\em root}. If $u$ and $v$ are two  nodes (possibly equal) of a rooted tree $(T,r)$ such that $u$ belongs to the path from $r$ to $v$, then $u$ is an ancestor of $v$ and $v$ a descendant of $u$ (with the convention that a node is an ancestor and a descendant of itself). 

\begin{definition}[Tree layout]
Let $G$ be a graph. A \emph{tree layout} of $G$ is a triple $\T=(T,r,\tau)$ where $(T,r)$ is a rooted tree and $\tau:V(G)\rightarrow V(T)$ is 
an {injective} mapping such that for every edge $\{x,y\}\in E(G)$, $\tau(x)$ is an ancestor of $\tau(y)$ in $(T,r)$ or vice-versa. We let $\mathcal{T}(G)$ denote the set of tree-layouts of $G$.
\end{definition}


Mimicking what we did with layouts, we can similarly define cost functions on tree-layouts. Given a graph $G$ and a  tree-layout $\T=(T,r,\tau)$ of $G$, such a function $\lambda_{G,\T}:V(T)\to\N$ takes as input a node $t$ of $T$ and outputs a non-negative integer.  Again we use $λ$ as a shortcut of  $\lambda_{G,\T}$, when the pair $G,\T$ is clear from the context.
Analogously to the definition of ${\bf p}_{\delta}(G,\L)$, we define  ${\bf p}_{\lambda}(G,\T)=\max\{\lambda(u)\mid u\in V(T)\}$. 
This yields the definition of the following general parameter:
$${\bf p}_{\lambda}(G)=\min\{{\bf p}_{\lambda}(G,\T) \mid \T\in \mathcal{T}(G)\}.$$

Suppose that $\T=(T,r,\tau)$ is a tree-layout of $G$. For every node $u$ of $T$, we let $T_u$ denote the subtree rooted at $u$ induced by the descendants of $u$ and define $X_{\T}(u)\subseteq V(G)$ as the set of vertices mapped to the nodes of $T_u$, that is $X_{\T}(u)=\{x\in V(G)\mid \tau(x) \mbox{ is a descendant of } u \mbox{ in } (T,r)\}$. Let us consider the following two cost functions defined on tree-layouts:
\begin{align*}
\lambda^{\sf v}(u) =  |N_{G}(X_{\T}(u))|,         & & \lambda^{\sf e}(u)  =  |E_{G}(X_{\T}(u))|.
\end{align*}

It is known that $\tw(G)={\bf p}_{\lambda^{\sf v}}(G)$ \cite{SeymourT93graph,DendrisKT97fugi}. We prove the pendant equality for edge tree-width. 

\begin{theorem}  $\etw={\bf p}_{\lambda^{\sf e}}$.
\end{theorem}
\begin{proof}
Let $G$ be a graph. 
We first prove that $\etw(G)\le{\bf p}_{\lambda^{\sf e}}(G)$.
Let $\T=(T,r,\rho)$ be a tree layout of $G$. 
Let also ${\sf L}=\langle x_1,\dots, x_n\rangle$ be a layout of $G$ obtained from a {DFS} ordering $\sigma$ of $(T,r)$: for two vertices $x$ and $y$ of $G$, we have $x\prec_{{\sf L}}y$ if and only if $\tau(x)\prec_{\sigma}\tau(y)$. 
Observe that by construction of ${\sf L}$, for every $i\in\{1, \dots, n\}$, we have $X_{\T}(\tau(x_i))\subseteq S_i$ and for every vertex $x_j\in S_i\setminus X_{\T}(\tau(x_i))$, $\tau(x_j)$ is neither an ancestor nor a descendant of $\tau(x_i)$ in $T$. It follows from the definition of a tree layout, that $C_G(S_i,x_i)\subseteq X_{\T}(\tau(x_i))$, proving the inequality.

\smallskip
We now prove that ${\bf p}_{\lambda^{\sf e}}(G)\le\etw(G)$.
From a layout ${\sf L}=\langle x_1, \dots, x_n\rangle$ of $G$, we recursively construct a rooted tree $(T,r)$ and an injective mapping $\tau: V(G)\rightarrow V(T)$ as follows. Initially $T$ is composed of a single (root) node $r$ and for every connected component $C$ of $G$, we create a child $u_C$ of $r$. If $x_C\in C$ is such that for every $y\in C$, $x_C\prec_{{\sf L}} y$ ($x_C$  is the smallest vertex of $C$ in ${\sf L}$), then we set $\tau(x_C)=u_C$. Then $(T,r)$ is completed by identifying for every connected component $C$ the node $u_c$ with the root of the rooted layout $\T_C=(T_C,r_C,\tau_C)$ of $G[C\setminus\{x_C\}]$ constructed from ${\sf L}[C\setminus\{x_C\}]$. Moreover, for every $x\in C\setminus\{x_C\}$, we set $\tau(x)=\tau_C(x)$. The resulting  triple $\T=(T,r,\tau)$ is clearly a tree-layout of $G$. Indeed if $\{x_i,x_j\}\in E(G)$ such that $x_i\prec_{{\sf L}} x_j$, then $x_j\in C_G(S_i,x_i)$ and thereby $\tau(x_i)$ is an ancestor of $\tau(x_j)$.

Let $x_j,x_i,x_h$ be three distinct vertices of $G$ such that $x_j\prec_{{\sf L}} x_i\prec_{{\sf L}} x_h$ and $\{x_j,x_h\}\in E_G(C_G(S_i,x_i))$. As $x_j\prec_{{\sf L}} x_h$ and ${\{x_j,x_h\}\in E(G)}$, $\tau(x_j)$ is an ancestor of $\tau(x_h)$. Moreover $\{x_j,x_h\}\in E_G(C_G(S_i,x_i))$ implies that $x_i$ and $x_h$ belongs to the same connected component of $G[S_j]$. As $x_i\prec_{{\sf L}} x_h$, this implies that $\tau(x_i)$ is an ancestor of $\tau(x_h)$. Finally, as $x_j\prec_{{\sf L}} x_i$, we obtain that $\tau(x_j)$ is an ancestor of $\tau(x_i)$, which we just argued, is an ancestor of $\tau( x_h)$. It follows that $\{x_j,x_h\}\in E_G(X_{\T}(\tau(x_i)))$, implying that ${\bf p}_{\lambda^{\sf e}}(G,\T)\le \max\{δ^{\sf ec}_{G,{\sf L}}(i)\mid i\in\{1, \dots, n\}\}$. 
\end{proof}

\paragraph{Relations in graphs.}

Given some graph parameter $\p$ and a partial ordering $\leq$ on graphs, we say that $\p$ is {\em closed under} $\leq$
if for every two graphs $H$ and $G$, it holds that $H\leq G\Rightarrow \p(H)\leq \p (G)$.

Given a vertex $v∈ V(G)$ of vertex-degree two and edge-degree two with neighbors $u$ and $w,$ we define the {\em dissolution} of $v$
to be the operation of deleting $v$ and adding the edge $\{u,w\}$ (here we agree that, in the case the edge $\{u,v\}$
already exists, we increase its multiplicity by one). Given two graphs $H,G,$ we say that $H$ is a {\em dissolution} of $G$
if $H$ can be obtained from $G$ after dissolving vertices of $G.$
A graph $H$ is a {\em topological minor} of a graph $G$,  denoted by $H\leq_{\sf tp} G$,  if $H$ is the dissolution of  some subgraph of $G$.

Contracting an edge $e=\{x,y\}$ in a graph $G$ results in the graph $G\textbackslash e$ such that $V(G\textbackslash e)=V(G)\setminus\{x,y\}\cup\{x_e\}$ and $\{u,v\}\in E(G\textbackslash e)$ if and only if either $u\neq x_e$, $v\neq x_e$ and $\{u,v\}\in E(G)$, or $v=x_e$ and either $\{x,u\}\in E(G)$ or $\{y,u\}\in E(G)$ (in the later case, we agree that the multiplicities of $\{x,u\}$ and $\{y,u\}$ are summed up so to define the multiplicity of $\{v,u\}$).
A graph $H$ is a {\em minor} of a graph $G$,  by $H\leq_{\sf mn} G$, if $H$ can be obtained from some subgraph of $G$ after a (possibly empty)  sequence of  edge contractions. Notice that if $H\leq_{\sf tp}G$, then $H\leq_{\sf mn}G$.

Two edge in a graph are called {\em incident} if they share some endpoint.
Given a graph $G$ and two incident edges $e_{1}=\{x,y\}$ and $e_{2}=\{x,z\}$ where $y\neq z$, the operation of {\em lifting}
the pair $e_1$, $e_{2}$ in $G$ removes the edges $e_{1}$ and $e_{2}$ and introduce the edge $\{y,z\}$ (in case $\{y,z\}$ already exists we increase by one its multiplicity). A graph $H$ is an {\em immersion} of a graph $G$,  by $H\leq_{\sf im} G$, if its can be obtained by a subgraph of $G$ after a sequence of incident edge lifting operations. Notice that if $H\leq_{\sf tp}G$, then $H\leq_{\sf im}G$.\medskip

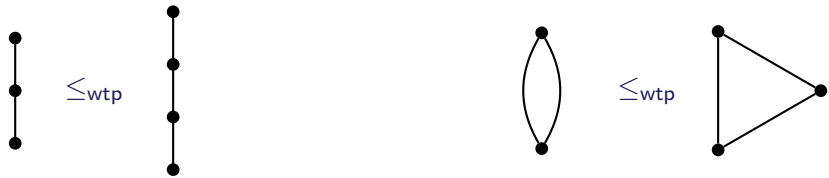
\begin{figure}[h]
\centering

\begin{tikzpicture}[thick,scale=0.7]
\tikzstyle{sommet}=[circle, draw, fill=black, inner sep=0pt, minimum width=4pt]

\draw  (-5,1)  node[sommet]{}
-- (-5,0)  node[sommet]{}
-- (-5,-1)  node[sommet]{};

\node[] (a) at (-3.5,0) {$\wtm$};

\draw  (-2,1.5)  node[sommet]{}
-- (-2,0.5)  node[sommet]{}
-- (-2,-0.5)  node[sommet]{}
-- (-2,-1.5)  node[sommet]{};

\begin{scope}[xshift=5cm]
\foreach \x/\y in {90/1,270/2}{
\draw node[sommet] (\y) at (\x:1.1){};
}

\draw[-,>=latex] (1) to[bend right] (2);
\draw[-,>=latex] (1) to[bend left] (2);
\node[] (a) at (2,0) {$\wtm$};
\end{scope}

\begin{scope}[xshift=9cm]

\foreach \x/\y in {120/1,240/2,360/3}{
\draw node[sommet] (\y) at (\x:1.3){};
}

\draw (1) -- (2) -- (3) -- (1) ;

\end{scope}
\end{tikzpicture}

\caption{Weak topological minor reduction rule. The path $P_3$ on three vertices is a weak topological minor of the path $P_4$. The cycle $C_2$ on two vertices is a weak topological minor of the cycle $C_3$.}
\label{weaktopminor}
\end{figure}

It is known that $\tw$ and $\pw$ are closed under minors and also closed under topological minors.
Also $\cw$ is closed under immersions.  Interestingly, edge-treewidth does not enjoy any of the closeness 
properties of treewidth, pathwidth, or cutwidth. 
For instance, in \autoref{counterex}, the graph $G_4$ is a topological minor of the graph $G_3$, proving that edge-treewidth is not closed under topological minor (and therefore neither closed under minors). Likewise, $G_2$ contains $G_1$ as an immersion, showing that edge-treewidth is not closed under immersion.

\begin{figure}[h]
\centering
\begin{center}
\begin{tikzpicture}[thick,scale=0.6]
\tikzstyle{sommet}=[circle, draw, fill=black, inner sep=0pt, minimum width=4pt]

\begin{scope}[xshift=0cm,yshift=-1cm]
\draw node[sommet] (1) at (-1.5,0){};
\draw node[sommet] (2) at (0,0){};
\draw node[sommet] (3) at (1.5,0){};
\draw node[sommet] (4) at (0,2){};

\draw [-] (1) to (4);
\draw [-] (2) to (4);
\draw [-] (3) to (4);
\draw[-,>=latex] (1) to[bend right=30] (2);
\draw[-,>=latex] (2) to[bend right=30] (1);
\draw[-,>=latex] (2) to[bend right=30] (3);
\draw[-,>=latex] (3) to[bend right=30] (2);

\node[] (a) at (0,-1.5) {$\etw(G_1)=3$};

\draw node[sommet] (a) at (0,-3){};
\draw node[sommet] (x) at (0,-4.5){};
\draw node[sommet] (b) at (-1,-6){};
\draw node[sommet] (c) at (1,-6){};
\draw[-,very thick,blue] (a) to (x);
\draw[-,very thick,blue] (x) to (b);
\draw[-,very thick,blue] (x) to (c);

\draw[-,>=latex] (a) to[bend right=30] (x);
\draw[-,>=latex] (a) to[bend right=30] (b);
\draw[-,>=latex] (a) to[bend left=30] (c);
\draw[-,>=latex] (x) to[bend right=20] (b);
\draw[-,>=latex] (x) to[bend right=40] (b);
\draw[-,>=latex] (x) to[bend left=20] (c);
\draw[-,>=latex] (x) to[bend left=40] (c);

\end{scope}

\begin{scope}[xshift=6cm,yshift=-1cm]
\draw node[sommet] (1) at (-1.5,0){};
\draw node[sommet] (2) at (0,0){};
\draw node[sommet] (3) at (1.5,0){};
\draw node[sommet] (4) at (0,2){};

\draw [-] (1) to (4);
\draw [-] (2) to (4);
\draw [-] (3) to (4);
\draw[-,>=latex] (1) to (2);
\draw[-,>=latex] (2) to (3);
\draw[-,>=latex] (1) to[bend right=30] (3);

\node[] (a) at (0,-1.5) {$\etw(G_2)=4$};

\draw node[sommet] (a) at (0,-3){};
\draw node[sommet] (x) at (0,-4.5){};
\draw node[sommet] (b) at (0,-6){};
\draw node[sommet] (c) at (0,-7.5){};
\draw[-,very thick,blue] (a) to (x);
\draw[-,very thick,blue] (x) to (b);
\draw[-,very thick,blue] (b) to (c);

\draw[-,>=latex] (a) to[bend right=30] (x);
\draw[-,>=latex] (x) to[bend right=30] (b);
\draw[-,>=latex] (x) to[bend left=30] (c);
\draw[-,>=latex] (a) to[bend left=40] (b);
\draw[-,>=latex] (b) to[bend right=40] (c);
\draw[-,>=latex] (c) to[bend left=40] (a);
\end{scope}

\begin{scope}[xshift=12cm]

\draw (0,0) circle (1.5) ;

\draw node[sommet] (0) at (0,0){};

\foreach \x/\y in {90/1,150/2, 210/3, 270/4, 330/5, 30/6}{
\draw node[sommet] (\y) at (\x:1.5){};
} 
\foreach \z in {2,4,6}{
\draw[-] (0) to (\z);
\draw[-,>=latex] (0) to[bend right=30] (\z);
\draw[-,>=latex] (\z) to[bend right=30] (0);
}

\node[] (a) at (0,-2.5) {$\etw(G_3)=6$};

\draw node[sommet] (ab) at (0,-4){};
\draw node[sommet] (bc) at (0,-5.2){};
\draw node[sommet] (ca) at (0,-6.4){};
\draw node[sommet] (x) at (0,-7.6){};
\draw node[sommet] (a) at (-1.5,-8.8){};
\draw node[sommet] (b) at (0,-8.8){};
\draw node[sommet] (c) at (1.5,-8.8){};
\draw[-,very thick,blue] (ab) to (bc);
\draw[-,very thick,blue] (bc) to (ca);
\draw[-,very thick,blue] (ca) to (x);
\draw[-,very thick,blue] (x) to (a);
\draw[-,very thick,blue] (x) to (b);
\draw[-,very thick,blue] (x) to (c);

\draw[-,>=latex] (ab) to[bend right=30] (a);
\draw[-,>=latex] (ab) to[bend right=30] (b);
\draw[-,>=latex] (bc) to[bend right=30] (b);
\draw[-,>=latex] (bc) to[bend left=30] (c);
\draw[-,>=latex] (ca) to[bend right=40] (a);
\draw[-,>=latex] (ca) to[bend left=40] (c);
\draw[-,>=latex] (a) to[bend left=20] (x);
\draw[-,>=latex] (a) to[bend left=40] (x);
\draw[-,>=latex] (a) to[bend left=60] (x);
\draw[-,>=latex] (b) to[bend right=20] (x);
\draw[-,>=latex] (b) to[bend right=40] (x);
\draw[-,>=latex] (b) to[bend right=60] (x);
\draw[-,>=latex] (c) to[bend right=20] (x);
\draw[-,>=latex] (c) to[bend right=40] (x);
\draw[-,>=latex] (c) to[bend right=60] (x);

\end{scope}

\begin{scope}[xshift=18cm]
\draw node[sommet] (0) at (0,0){};

\foreach \x/\y in {150/2, 270/4,  30/6}{
\draw node[sommet] (\y) at (\x:1.5){};
} 
\foreach \z in {2,4,6}{
\draw[-] (0) to (\z);
\draw[-,>=latex] (0) to[bend right=30] (\z);
\draw[-,>=latex] (\z) to[bend right=30] (0);
}
\draw[-,>=latex] (2) to[bend right=30] (4);
\draw[-,>=latex] (4) to[bend right=30] (6);
\draw[-,>=latex] (6) to[bend right=30] (2);

\node[] (a) at (0,-2.5) {$\etw(G_4)=8$};

\draw node[sommet] (a) at (0,-4){};
\draw node[sommet] (x) at (0,-5.5){};
\draw node[sommet] (b) at (0,-7){};
\draw node[sommet] (c) at (0,-8.5){};
\draw[-,very thick,blue] (a) to (x);
\draw[-,very thick,blue] (x) to (b);
\draw[-,very thick,blue] (b) to (c);

\draw[-,>=latex] (a) to[bend right=30] (x);
\draw[-,>=latex] (a) to[bend right=50] (x);
\draw[-,>=latex] (a) to[bend right=70] (x);
\draw[-,>=latex] (x) to[bend right=30] (b);
\draw[-,>=latex] (x) to[bend right=50] (b);
\draw[-,>=latex] (x) to[bend right=70] (b);
\draw[-,>=latex] (x) to[bend left=30] (c);
\draw[-,>=latex] (x) to[bend left=50] (c);
\draw[-,>=latex] (x) to[bend left=70] (c);
\draw[-,>=latex] (a) to[bend left=50] (b);
\draw[-,>=latex] (b) to[bend right=50] (c);
\draw[-,>=latex] (c) to[bend left=50] (a);
\end{scope}

\end{tikzpicture}
\end{center}
\caption{The graph $G_1$ contains the graph $G_2$ as an immersion, while the graph $G_3$ contains the graph $G_4$ as a topological minor. The values of $\etw(G_1$, $\etw(G_2)$, $\etw(G_3)$ and $\etw(G_4)$ are certified by the tree-layouts (in blue) drawn below the respective graphs. The lower bounds on the edge-treewidth of the above graphs have been verified by exhaustively considering all possible layouts (using a computer program).
}
\label{counterex}
\end{figure}
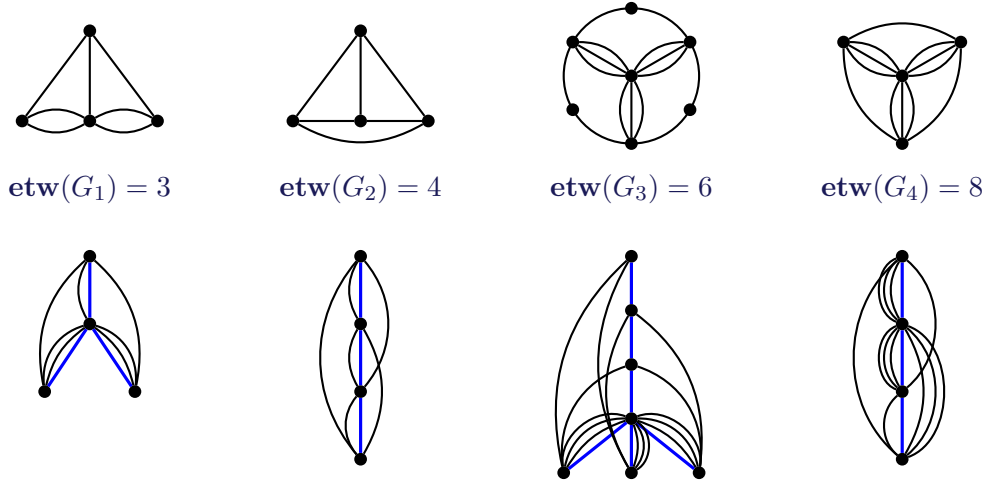
\medskip

 Our next step is to introduce a new partial ordering relation on graphs and prove that $\etw$ is closed in this new relation.
A graph $H$ is a \emph{weak topological minor} of a graph $G$, denoted by $H\wtm G$, if $H$ is obtained from a subgraph of $G$ by contracting edges whose both endpoints have edge-degree two and vertex-degree two (see \autoref{weaktopminor}). We observe that the 2-cycle is a weak topological minor of the $3$-cycle (and henceforth of every chordless cycle).

\begin{theorem} \label{th_weak_topological_minor}
Edge-treewidth is closed under taking weak topological minors.
\end{theorem}

\begin{proof}
Let $G=(V,E)$ be a graph and $e=\{x,y\}$ be an edge such that $x$, $y$ are two vertices each of vertex-degree two and edge-degree two. Let $\T=(T,r,\tau)$ be a tree layout of $G$. As $\{x,y\}\in E$, we can assume, without loss of generality, that $\tau(x)$ is an ancestor of $\tau(y)$. A tree layout $\T'=(T,r,\tau')$ of $G\textbackslash e$ is obtained from $\T$ as follows. Let $x_{e}$ be the vertex resulting from the contraction of $e$. For every vertex $z\in V(G\textbackslash e)$ such that $z\neq x_e$, we set $\tau'(z)=\tau(z)$, and $\tau'(x_e)=\tau(x)$. 

Let us first argue $\T'=(T,r,\tau')$ is a tree layout of $G\textbackslash e$. By construction, we have that for every edge $\{a,b\}$ not incident to $x_e$, either $\tau'(a)$ is an ancestor of $\tau'(b)$ or vice versa. Let $w$ be the unique neighbor of $x$ distinct from $y$ in $G$ and $z$ be the unique neighbor of $y$ distinct from $x$ in $G$. Observe that in $G\textbackslash e$, the neighbors of $x_e$ are $z$ and $w$. As in $\T$, $\tau(w)$ is an ancestor of $\tau(x)$ or vice versa and as $\tau'(x_e)=\tau(x)$, we have that $\tau'(w)$ is an ancestor of $\tau'(x_e)$ or vice versa. Suppose that $\tau(z)$ is a descendant of $\tau(y)$, then it is also a descendant of $\tau(x)$, implying that $\tau'(z)$ is a descendant of $\tau'(x_e)$. If on the contrary, $\tau(z)$ is an ancestor of $\tau(y)$, then it is either a descendant or an ancestor of $\tau(x)$. This, in turns, implies that $\tau'(z)$ is either a descendant or an ancestor of $\tau'(x_e)$.

It remains to prove that $\ewidth(G,\T')\le \ewidth(G,\T)$. Recall that contracting $e=\{x,y\}$ amounts to removing the vertex $y$ and its incident edges $\{x,y\}$ and $\{y,z\}$, identifying vertex $x$ with $x_e$, adding the edge $\{z,x_e\}$.  Let $P$ the smallest subpath of $T$ containing $\tau(x)$, $\tau(y)$ and $\tau(z)$. Observe that, by construction, for every node $u$ not in $P$ we have $E_G(X_{\T'}(u))=E_G(X_{\T}(u))$. So assume that $u$ is a node of $P$. From the previous paragraph we know that one of the following three cases holds:
\begin{enumerate}
\item $\tau(x)$ is an ancestor of $\tau(y)$ that is an ancestor of $\tau(z)$: if $u$ is an ancestor of $\tau(y)$, then we have that $E_G(X_{\T'}(u))=E_G(X_{\T}(u))\setminus\{\{x,y\}\}\cup \{\{z,x_e\}\}$. Otherwise we have ${E_G(X_{\T'}(u))=E_G(X_{\T}(u))\setminus\{\{y,z\}\}\cup \{\{z,x_e\}\}}$.
\item $\tau(x)$ is an ancestor of $\tau(z)$ that is an ancestor of $\tau(y)$: if $u$ is an ancestor of $\tau(z)$, then we have that $E_G(X_{\T'}(u))=E_G(X_{\T}(u))\setminus\{\{x,y\}\}\cup \{\{z,x_e\}\}$. Otherwise we have ${E_G(X_{\T'}(u))=E_G(X_{\T}(u))\setminus\{\{x,y\},\{y,z\}\}}$.
\item $\tau(z)$ is an ancestor of $\tau(x)$ that is an ancestor of $\tau(y)$: if $u$ is an ancestor of $\tau(x)$, then we have that $E_G(X_{\T'}(u))=E_G(X_{\T}(u))\setminus\{\{y,z\}\}\cup \{\{z,x_e\}\}$. Otherwise we have $E_G(X_{\T'}(u))=E_G(X_{\T}(u))\setminus\{\{x,y\},\{y,z\}\}$.
\end{enumerate}
It follows that for every node $u$ of $T$, we have that 
$\lambda^{\sf e}_{G,\T'}(u)\le \lambda^{\sf e}_{G,\T}(u)$, concluding the proof.
\end{proof}

%
%
%
%

\section{Obstructions}
\label{uiower}

Given a graph class ${\cal G}$ and a partial relation $\leq$, we define the obstruction of ${\cal G}$ with respect to $\leq$, denoted by 
$\obs_{\leq}({\cal G})$, as the set of all $\leq$-minimal graphs not in ${\cal G}$. Clearly, if ${\cal G}$ is closed under $\leq$, then 
$\obs_{\leq}({\cal G})$ can be seen as a complete characterisation of ${\cal G}$, as $G\in{\cal G}$ iff $\forall H\in\obs({\cal G})\ H\not\leq G$ 
Hereafter, for $i\in \Bbb{N},$ we denote by $\obs_i$  the weak-topological minor obstruction set of the family of graphs of edge treewidth at most $i$, in order words $\obs_i=\obs_{\wtm}(\{G\mid \etw(G)\leq i\})$.

We next prove that $\obs_3$ is infinite, while $\obs_1$ and $\obs_2$ are finite and respectively characterize forests and cactus graphs.

\begin{theorem} \label{th_etw_forest}
For a graph $G$, the following properties are equivalent: 
\begin{enumerate}
\item  ${\etw}(G)\le 1$; 
\item $G$ is a forest; 
\item $\obs_1=\{C_2\}$.
\end{enumerate}
\end{theorem}
\begin{proof}
\emph{(3. $\Rightarrow$ 2.}) 
Suppose that $G$ contains a cycle. Then we have $C_2\wtm G$. Finally, as $\etw(C_2)=2$, \autoref{th_weak_topological_minor} implies $\etw(G)\ge 2$.

\noindent
\emph{(2. $\Rightarrow$ 1.}) Suppose that $G$ is a forest.  Consider the tree-layout $(T,r,\tau)$ where $T$ is obtained from $G$ by adding a root $r$ adjacent to an arbitrary vertex of every component of $G$ and where $\tau$ maps every vertex of $G$ to its copy in $T$. Clearly we have that $\ewidth(G,\T)=1$.

\noindent
\emph{(1. $\Rightarrow$ 3.}) Observe that $C_2$ is minimal for the weak topological minor relation and that $\etw(C_2)=2$. This implies that $C_2\in\obs_1$. Suppose that $\obs_1$ contains a graph $H$ distinct from $C_2$. It follows that $H$ excludes every $C_2$ as a weak topological minor. But then $H$ is a forest, implying that $\etw(H)\le 1$: contradiction.
\end{proof}

A graph in which that every edge belongs to at most one cycle is called a \emph{cactus graph}~\cite{HaynesHS98Fundamentals}. Equivalently, $G$ is a cactus graph if and only if its blocks (biconnected components) are cycles.
For a  graph $G$, ${\sf bc}(G)$ denote the set of its biconnected components, also called \emph{blocks}, and ${\sf cv}(G)$ denote the set of cut vertices of $G$. We define $B_G$ as the graph in which the vertex set is one-to-one mapped to ${\sf bc}(G)\cup{\sf cv}(G)$ and two vertices $x$ and $y$ of $B_G$ are adjacent if and and one if $x$ is mapped to a block $B$ and $y$ to a cut vertex belonging to $B$. Observe that $B_G$ is a tree where the leaves correspond to blocks of $G$. We call $B_G$ the \emph{block tree} of $G$.

\begin{theorem} \label{th_cactus}
For a graph $G$, the following properties are equivalent:
\begin{enumerate}
\item ${\etw}(G) \le 2$;
\item $G$ is a cactus graph;
\item $\obs_2=\{Z_2^i\mid 1\le i\le 4\}$  (see Figure \ref{etw2}).
\end{enumerate}
\end{theorem}

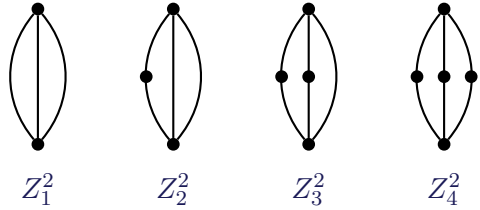
\begin{figure}[h]
\centering
\begin{center}
\begin{tikzpicture}[thick,scale=0.6]
\tikzstyle{sommet}=[circle, draw, fill=black, inner sep=0pt, minimum width=4pt]

\begin{scope}[xshift=0cm]
\foreach \x/\y in {-1.5/1,1.5/2}{
\draw node[sommet] (\y) at (0,\x){};
}

\draw[-,>=latex] (1) to[bend right=40] (2);
\draw[-,>=latex] (1) to[bend left=0] (2);
\draw[-,>=latex] (1) to[bend left=40] (2);

\node[] (a) at (0,-2.5) {$Z_1^2$};

\end{scope}

\begin{scope}[xshift=3cm]
\foreach \x/\y in {-1.5/1,1.5/2}{
\draw node[sommet] (\y) at (0,\x){};
}

\draw node[sommet] (3) at (-0.6,0){};

\draw[-,>=latex] (1) to[bend right=40] (2);
\draw[-,>=latex] (1) to[bend left=0] (2);
\draw[-,>=latex] (1) to[bend left=40] (2);

\node[] (a) at (0,-2.5) {$Z_2^2$};
\end{scope}

\begin{scope}[xshift=6cm]
\foreach \x/\y in {-1.5/1,1.5/2}{
\draw node[sommet] (\y) at (0,\x){};
}

\draw node[sommet] (3) at (-0.6,0){};
\draw node[sommet] (4) at (0,0){};

\draw[-,>=latex] (1) to[bend right=40] (2);
\draw[-,>=latex] (1) to[bend left=0] (2);
\draw[-,>=latex] (1) to[bend left=40] (2);

\node[] (a) at (0,-2.5) {$Z_3^2$};
\end{scope}

\begin{scope}[xshift=9cm]
\foreach \x/\y in {-1.5/1,1.5/2}{
\draw node[sommet] (\y) at (0,\x){};
}

\draw node[sommet] (3) at (-0.6,0){};
\draw node[sommet] (4) at (0,0){};
\draw node[sommet] (5) at (0.6,0){};

\draw[-,>=latex] (1) to[bend right=40] (2);
\draw[-,>=latex] (1) to[bend left=0] (2);
\draw[-,>=latex] (1) to[bend left=40] (2);

\node[] (a) at (0,-2.5) {$Z_4^2$};
\end{scope}

\end{tikzpicture}
\end{center}
\caption{The weak topological minor obstruction set $\obs_2=\{Z_1^2,Z_2^2,Z_3^2,Z_4^2\}$ for the graphs of edge-treewidth at most $2$.
}
\label{etw2}
\end{figure}

\begin{proof}
\emph{(3. $\Rightarrow$ 2.}) 
Let $G$ be a graph that is not a cactus graph. Then $G$ has an edge $\{x,y\}$ that is contained in two cycles, say $C_1$ and $C_2$. This implies that the subgraph $G[C_1\cup C_2]$ contains two vertices $u$ and $v$ and is composed of three edge-disjoint paths between $u$ and $v$. It follows that $G[C_1\cup C_2]$, and thereby $H$, contains one of the graphs $Z_i^2\in\obs_2$ ($1\le i\le 4$) as a weak topological minor.

\medskip
\noindent
\emph{(2. $\Rightarrow$ 1.})
Let $G$ be a cactus graph. As every block $B$ of $G$ is either an edge or a cycle, for every vertex $x\in B$, there exists a layout ${\sf L}(B,x)$ of $B$ starting at $x$ such that $\lwidth(B,{\sf L})\le 2$.

We recursively construct  $\T=(T,r,\tau)$ from the block tree $B_G$ of $G$. We root $B_G$ at an arbitrary leaf. Let $B_r$ be the block corresponding to that leaf. Observe that for every block $B$ distinct from $B_r$, the parent of its corresponding vertex in $B_G$ is mapped to a cut vertex $x_B\in B$.
We choose an arbitrary vertex $x_r$ of $B_r$ and define $\tau(x_r)$ as the root $r$ of $T$. Then for every vertices $x$ and $y$ of $B_r$, $\tau(y)$ is a child of $\tau(x)$ if and only if $x$ immediately precedes $y$ in ${\sf L}(B_r,x_r)$. Suppose that $B$ is a block of $G$ containing a cut vertex $x_B$ such that $x_B$ is the unique vertex of $B$ with $\tau(x_B)$ been defined. Then for every vertices $x$ and $y$ of $B$, $\tau(y)$ is a child of $\tau(x)$ if and only if $x$ immediately precedes $y$ in ${\sf L}(B,x_B)$. This clearly defines a tree-layout of $G$. As every edge of $G$ belongs to some block $B$ of $G$ and as for every block $\lwidth({\sf L},B)\le 2$, we have that for every node $u$ of $T$, $\lambda^{\sf e}(G,\T,u)\le 2$, implying that $\etw(G)\le 2$.

\medskip
\noindent
 \emph{(1. $\Rightarrow$ 3.}) 
We first prove that $\{Z_2^i\mid 1\le i\le 4\}\subseteq \obs_2$. Observe that none of these graphs has an edge incident to two vertices each of vertex-degree  and edge-degree two. So the graphs of $\obs_2$ are minimal for the weak topological minor relation. Consider a layout ${\sf L}\in\mathcal{L}(G)$ such that $\etw(G)= \lwidth({\sf L},G)$. Suppose that $G$ contains some graph $H\in \obs_2$ as a weak topological minor. Then $G$ contains two vertices, say $x_i$ and $x_j$ with $x_i\prec_{{\sf L}} x_j$, and $3$ edge-disjoint paths $P_1$, $P_2$ and $P_3$ between $x_i$ and $x_j$. Then $E_G(C_G(S_j,x_j))$ contains at least one edge from each of $P_1$, $P_2$ and $P_3$. This implies that $\etw(G)\ge 3$. 

Suppose for the sake of contradiction that there exists a graph $H\in \obs_2\setminus \{Z_2^i\mid 1\le i\le 4\}$. It follows that $H$ excludes every $Z_2^i$ ($1\le i\le 4)$ as a weak topological minor. But then $H$ is a cactus graph, implying that $\etw(H)\le 2$ and thereby $H\not\in\obs_2$.
\end{proof}

\begin{lemma}
\label{kolop}
The obstruction set  $\obs_3$ is infinite. \end{lemma}

\begin{figure}[h]
\begin{center}
\begin{tikzpicture}[thick,scale=0.7]
\tikzstyle{sommet}=[circle, draw, fill=black, inner sep=0pt, minimum width=4pt]

\begin{scope}[xshift=0cm]
\foreach \x/\y in {90/1,270/2}{
\draw node[sommet] (\y) at (\x:1.1){};
}

\draw[-,>=latex] (1) to[bend right=20] (2);
\draw[-,>=latex] (1) to[bend right=60] (2);
\draw[-,>=latex] (1) to[bend left=20] (2);
\draw[-,>=latex] (1) to[bend left=60] (2);

\node[] (a) at (0,-2) {$Z_2^3$};

\end{scope}

\begin{scope}[xshift=2.8cm]
\foreach \x/\y in {120/1,240/2,360/3}{
\draw node[sommet] (\y) at (\x:1.3){};
}

\draw[-,>=latex] (1) to[bend right=20] (2);
\draw[-,>=latex] (1) to[bend left=20] (2);
\draw[-,>=latex] (2) to[bend right=20] (3);
\draw[-,>=latex] (2) to[bend left=20] (3);
\draw[-,>=latex] (3) to[bend right=20] (1);
\draw[-,>=latex] (3) to[bend left=20] (1);

\node[] (a) at (0,-2) {$Z_3^3$};
\end{scope}

\begin{scope}[xshift=6.8cm]
\foreach \x/\y in {0/1,90/2,180/3,270/4}{
\draw node[sommet] (\y) at (\x:1.3){};
}

\draw[-,>=latex] (1) to[bend right=20] (2);
\draw[-,>=latex] (1) to[bend left=20] (2);
\draw[-,>=latex] (2) to[bend right=20] (3);
\draw[-,>=latex] (2) to[bend left=20] (3);
\draw[-,>=latex] (3) to[bend right=20] (4);
\draw[-,>=latex] (3) to[bend left=20] (4);
\draw[-,>=latex] (4) to[bend right=20] (1);
\draw[-,>=latex] (4) to[bend left=20] (1);

\node[] (a) at (3,0) {\dots};
\node[] (a) at (0,-2) {$Z_4^3$};
\end{scope}

\begin{scope}[xshift=12cm]
\foreach \x/\y in {0/1,1.5/2,3/3,7/4,8.5/5}{
\draw node[sommet] (\y) at (\x,0){};
}

\draw[-,>=latex] (1) to[bend right=20] (2);
\draw[-,>=latex] (1) to[bend left=20] (2);
\draw[-,>=latex] (2) to[bend right=20] (3);
\draw[-,>=latex] (2) to[bend left=20] (3);
\draw[-,>=latex] (4) to[bend right=20] (5);
\draw[-,>=latex] (4) to[bend left=20] (5);
\draw[-,>=latex] (1) to[bend right=20] (5);
\draw[-,>=latex] (1) to[bend left=20] (5);


\node[] (a) at (4.25,0) {\dots};
\node[] (b) at (5.75,0) {\dots};
\node[] (a) at (4.25,-2) {A layout of $Z_n^3$};
\end{scope}

\end{tikzpicture}
\end{center}
    \caption{The set $\{Z_2^3, Z_3^3, \dots, Z_n^3,\dots \}\subseteq\obs_3$ forms an infinite antichain for the weak topological minor relation.}
    \label{layoutw4}
\end{figure}
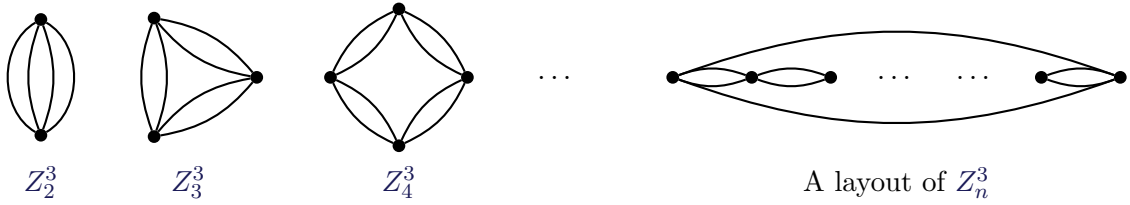
\begin{proof}
We define the graph $Z_n^3$ as the graph obtained from the cycle $C_n$ by duplicating every edge once (see Figure \ref{layoutw4}).
We observe that the set $\{Z_n^3\mid n\ge 2\}$ forms an infinite antichain with respect to $\wtm$ and that for every $n\ge 2$, $\etw(Z_n^3)=4$. First, as for each of these graphs every vertex has degree $4$, they are minimal for the weak topological minor relation. Second, let ${\sf L}$ be any layout of $Z_n^3$ (for $n\ge 2$). We remark that for $u$, the last vertex in ${\sf L}$, we have $\delta^{\sf ce}_{G,{\sf L}}(u)=4$. 
Finally, observe that if ${\sf L}=\langle x_1,\dots, x_n\rangle$ is a layout of $Z_3^n$ such that for every $i\in[2,n]$, $x_i$ is adajcent to $x_{i-1}$, then $\lwidth(G,{\sf L})=4$. It follows that $\{Z_n^3\mid n\ge 2\}\subseteq \mathcal{Z}_3$.
\end{proof}

%
%
%

\section{A  parametric equivalence}
\label{nioklop}

Our next step is to show that edge-treewidth can be parametrically expressed using the maximum edge-degree parameter.
For this we will define a new parameter, using the edge-degree as basic ingredient, and we will prove its parameteric equivalence with edge-treewidth.

\begin{lemma} \label{lem_biconnected_rooted_layout}
Let $G=(V,E)$ be a biconnected graph. For every vertex $u$ of $G$, we have $\etw(G,u) \le \etw(G)^2+2\cdot\etw(G)$, where $\etw(G,u)=\max\{\lwidth(G,{\sf L}) \:|\: {\sf L} \in \mathcal{L}(G) \land {\sf L}{(1)}= u\}$. 
\end{lemma}
\begin{proof}
 Let ${\sf L}=\langle x_1, \dots, x_n \rangle$ be a layout of $G$. Suppose that $u=x_i$. Consider the layout ${\sf L}' = \langle x_i, x_1, \dots, x_{i-1}, x_{i+1}, \dots x_n \rangle$. We  observe that for every vertex $x\in \{x_{i+1},\dots,x_{n}\}$, we have $\delta^{\sf ce}_{G,{\sf L}}(x)=\delta^{\sf ce}_{G,{\sf L}'}(x)$. However, for a vertex $x\in \{x_1,\dots, x_{i-1}\}$, we have that $\delta^{\sf ce}_{G,{\sf L}'}(x)\le \delta^{\sf ce}_{G,{\sf L}}(x)+\ell$, where $\ell$ is upper bounded by the degree of $x_i$ in $G$.
 
First observe that $|N(x_i)\cap \{x_1,\dots, x_{i-1}\}|$ is at most $\delta^{\sf ce}_{G,{\sf L}}(x_i)\le \lwidth(G,{\sf L})$. It remains to bound $|N(x_i)\cap S_{i+1}|$. Let $H$ be the subgraph of $G$ induced by $C_G(S_i,x_i)$. As $G$ is biconnected, every connected component of $H-x_i$ has at least one neighbor that appears prior to $x_i$ in ${\sf L}$. It follows that the number of connected components of $H-x_i$ is at most $\delta^{\sf ce}_{G,{\sf L}}(x_i)\le \lwidth(G,{\sf L})$. Let $C$ be one of these connected components and $x_C$ be the first vertex of $C$ in ${\sf L}$. Then observe that  $|N(x_i)\cap C|\le \delta^{\sf ce}_{G,{\sf L}}(x_C)\le \lwidth(G,{\sf L})$.
It follows that $|N(x_i)\cap S_{i+1}|\le \lwidth(G,{\sf L})^2$, and thereby $\ell\le  \lwidth(G,{\sf L})^2+ \lwidth(G,{\sf L})$, proving the result.
\end{proof}

\begin{theorem} \label{th_biconnected}
For every graph $G$, we have
\blue{${\etw}(G) \le \max\{{\etw}(B)^2+\etw(B) \:|\: B\in {\sf bc}(G)\}$.}
\end{theorem}
\begin{proof}
We compute a tree-layout $\T=(T,r,\tau)$ from the block tree $B_G$ of $G$ as described in the proof of Theorem~\ref{th_cactus} ($2.\Rightarrow 1.$). This construction defines for every block $B$ of $G$ a root vertex $x_B$. Let   ${\sf L}_B$ denote the layout starting at $x_B$ used in the construction of $\T$. Clearly, we have 
$\ewidth(G,\T)\le \max\{\lwidth(B,{\sf L}_B) \:|\: B\in{\sf bc}(G)\}$. It follows from Lemma~\ref{lem_biconnected_rooted_layout} that $\ewidth(G,\T)\le \max\{{\etw}(B)^2+2\cdot\etw(B) \:|\: B\in {\sf bc}(G)\}$.
\end{proof}
%

\begin{theorem} \label{th_equivalence}
For every graph $G$, we have
\[{\etw}(G) \sim \max\big\{\{\Delta_{\sf e}(B) \:|\: B \in {\sf bc}(G)\} \cup \{ {\tw}(B) \:|\: B \in {\sf bc}(G)\}\big\}, \]
where $\Delta_{\sf e}(G)$ stands for the maximum edge-degree in $G$.
\end{theorem}

\begin{proof}
For a graph $G=(V,E)$, we define ${\bf p}(G) = \max\{\{\Delta_{\sf e}(B) \:|\: B \in {\sf bc}(G)\} \cup \{ {\tw}(B) \:|\: B \in {\sf bc}(G)\}\}$.

Let us first prove that ${\etw}(G) \le {\bf p}(G)^4+{\bf p}(G)^2$. It is known that $\tw={\bf p}_{\delta^{\sf vc}}$~\cite{DendrisKT97fugi} and we defined $\etw={\bf p}_{δ^{\sf ec}}$. It follows from the definition of ${\bf p}_{\delta^{\sf vc}}$ and ${\bf p}_{δ^{\sf ec}}$ that, for every graph $H$, $\etw(H)\le {\tw}(H)\cdot \Delta_{\sf e}(H)$. Applying this inequalities to the blocks of $G$, \autoref{th_biconnected} implies that $\etw(G)\le \max\{{\tw}(B)^2\cdot\Delta_{\sf e}(B)^2+2\cdot\tw(B)\cdot\Delta_{\sf e}(B) \:|\: B\in {\sf bc}(G)\}$, proving the upper bound.

Let us now prove that ${\bf p}(G)\le \etw(G)$. It is trivially the case if ${\bf p}(G)=\max\{{\tw}(B) \:|\: B \in {\sf bc}(G)\}$. Indeed, we then have ${\bf p}(G)\le \tw(G)\le \etw(G)$. So assume that ${\bf p}(G)=\max\{\Delta_{\sf e}(B) \:|\: B \in {\sf bc}(G)\}$. Let $\T=(T,r,\tau)$ be a tree layout of $G$ such that $\etw(G)=\ewidth(G,\T)$. Let $B_m$ be the block of $G$ such that $\Delta_{\sf e}(B_m)={\bf p}(G)$. Let $x_m$ be a vertex of $B_m$ with maximum edge-degree in $B$ and $N$ be its set of neighbors in $B$. 
We have different cases to consider depending on the position of node $u_m=\tau(x_m)$ in $\T$.
\begin{itemize}

\item Suppose that there exists in $T$ at most one child $v$ of the node $u_m$ such that $X_{\T}(v)\cap N\neq\emptyset$. 
Observe that we have  $\ewidth(G,\T,v)\ge |X_{\T}(v)\cap N|$  and $\ewidth(G,\T,u_m)\ge \Delta_{\sf e}(B_m)-|X_{\T}(v)\cap N|$. This implies that $\Delta_{\sf e}(B_m)/2={\bf p}(G)/2\le \etw(G)$.

\item Let $v_1,\dots v_{\ell}$, with $\ell>1$, be the children of $u_m$ in $T$ such that for every $i\in[\ell]$, $X_{\T}(v_i)\cap N\neq\emptyset$. Let $n_i$ be the number of vertices of $N$ in $X_{\T}(v_i)$. Observe that for every $i\in\ell$, we have $\ewidth(G,\T,v_i)\ge n_i$. Moreover as $B_m$ is a block, $x_m$ is not a cut vertex of $B_m$. Thereby for every pair $i$, $j$, $i\neq j$, there must exist a path from every vertex $x_i\in X_{\T}(v_i)\cap N$ to every vertex of $x_j\in X_{\T}(v_j)\cap N$ avoiding $x_m$. As $\T$ is a tree layout of $G$, such a path contains a vertex $z$ such that $\tau(z)$ is an ancestor of $u_m$. It follows that $\ewidth(G,\T,u_m)\ge \ell+\Delta_{\sf e}(B_m)-\sum_{i\in[l]} n_i$. In other words, we have that $\ewidth(G,\T)\ge \max \{n_1,\dots, n_{\ell}, \ell+\Delta_{\sf e}(B_m)-\sum_{i\in[l]} n_i\}$.
We observe that $\max \{n_1,\dots, n_{\ell}, \ell+\Delta_{\sf e}(B_m)-\sum_{i\in[l]} n_i\}\ge \sqrt{\Delta_{\sf e}(B_m)}$. This lower bound is attained when $\ell=\sqrt{\Delta_{\sf e}(B_m)}-1$ and for every $i\in[\ell]$, $n_i=\sqrt{\Delta_{\sf e}(B_m)}$. This implies that $\sqrt{{\bf p}(G)}\le \etw(G)$.
\end{itemize}
So we proved that 
$$\sqrt{{\bf p}(G)} \le \etw(G)\le {\bf p}(G)^4+2\cdot{\bf p}(G)^2.\vspace{-7mm}$$
\end{proof}
\medskip

An algorithmic consequence of \autoref{th_equivalence} and the linear FPT 5-approximation  algorithm for treewidth in~\cite{BodlaenderDDFLP13Anoca} is the following.

\begin{corollary}
\label{iocolorli}
One can construct a linear algorithm that, given a graph $G$ 
and a non-negative integer $k$, either returns a layout of $G$ certifying that 
$\etw(G)=O(k^{4})$ or reports that $\etw(G)>k$. 
\end{corollary}

It remains an open question whether there is an {\sf FPT}-algorithm checking whether $\etw(G)\leq k$.

\section{Universal obstructions for edge-treewidth}
\label{opertlsok}

\paragraph{Walls, thetas, and fans.}\label{label_mistreatment}
Let  $k,r∈\Bbb{N}.$ The
\emph{$(k\times r)$-grid} is the
graph whose vertex set is $[k]\times[r]$ and two vertices $(i,j)$ and $(i',j')$ are adjacent if and only if $|i-i'|+|j-j'|=1.$
For $i\geq 1$, we use $Γ_{i}$ for the $(i\times i)$-grid.
An  \emph{$r$-wall}, for some integer $r≥ 2,$ is the graph, denoted by $W_{i}$, obtained from a
$(2 r\times r)$-grid
with vertices $(x,y)
	∈[2r]\times[r],$
after the removal of the
``vertical'' edges $\{(x,y),(x,y+1)\}$ for odd $x+y,$ and then the removal of
all vertices of degree one.

\begin{figure}[h]
\begin{center}
\begin{tikzpicture}[thick,scale=0.6]
\tikzstyle{sommet}=[circle, draw, fill=black, inner sep=0pt, minimum width=3.5pt]

\begin{scope}[]

\foreach \x in {-5,-3,,-1,1,3,5}{
\draw node[sommet] () at (\x,0){};
}
\foreach \x in {-5,-3,,-1,1,3,5}{
\draw node[sommet] () at (\x,5){};
}
\foreach \x in {-6,-5,-4,-3,-2,-1,0,1,2,3,4,5}{
\foreach \y in {1,2,3,4}
\draw node[sommet] () at (\x,\y){};
}

\foreach \y in {0,5}{
\draw (-5,\y) -- (5,\y);
}
\foreach \y in {1,2,3,4}{
\draw (-6,\y) -- (5,\y);
}
\foreach \x in {-5,-3,-1,1,3,5}{
\foreach \y in {0,2,4}{
\draw (\x,\y) -- (\x,\y+1);
}
}
\foreach \x in {-6,-4,-2,0,2,4}{
\foreach \y in {1,3}{
\draw (\x,\y) -- (\x,\y+1);
}
}
\end{scope}

\end{tikzpicture}
\end{center}
    \caption{The wall $W_{6\times 6}$.}
    \label{fig_wall}
\end{figure}
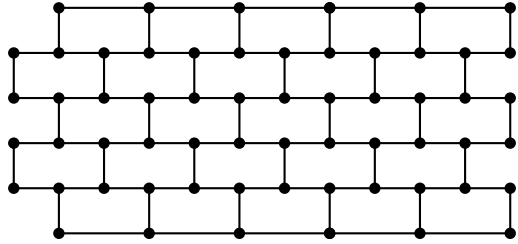

The $i$-\emph{theta}, denoted $\theta_i$, for $i\in\mathbb{N}_{\geq 1}$, is the graph composed by two vertices, called \emph{poles}, and $i$ parallel edges between them. The \emph{$i$-fan}, denoted $\varphi_i$ for $i\in\mathbb{N}_{\geq 1}$, is the graph obtained if we take a path $P_i$ on $i$ vertices plus a universal vertex, i.e., a new vertex adjacent with all the  vertices of the path $P_{i}$. The end vertices of the path are called \emph{extreme vertices} of $φ_{i}$.\medskip


\paragraph{Universal obstructions.}
By the grid minor theorem~\cite{RoberstonS84GMIII,RobertsonST94quick} (see also~\cite{Die05graph}), we know that
the treewidth is parametrically equivalent to the size of the largest wall contained as a topological minor.
That way, we can claim that the walls as a set of ``universal minor-obstructions'' for treewidth. Before proving a similar result for edge-treewidth and the weak topological minor relation, let us formalize the notion \emph{universal obstructions}.
\medskip

Let $\leq$ be some partial ordering relation.  Given a set of graphs ${\cal A}$ and a 
graph $G$, we say that ${\cal A}\leq G$ if $\exists H\in {\cal A}\ H\leq G$.
A {\em parameterized set of graphs} is a set ${\cal H}=\{H_{i}\mid i\in \Bbb{N}\}$ of graphs, indexed by non-negative integers.
Given an partial ordering relation $\leq$, we say that ${\cal H}$ is {\em $\leq$-monotone} if for every  $i\in\Bbb{N}$, $H_i\leq { H}_{i+1}$.  
Let $\mathcal{H}^1$ and $\mathcal{H}^2$ be two $\leq$-{monotone}  parameterized sets of graphs. We say that ${\cal H}^1\leq {\cal H}^2$ if there is a function $f:\Bbb{N}\to\Bbb{N}$ such that for every $i\in\Bbb{N}$, $H_{i}^1\leq H_{f(i)}^{2}$.
We say that ${\cal H}^{1}$ and ${\cal H}^{2}$ are \emph{$\leq$-equivalent} if  ${\cal H}^1\leq {\cal H}^2$  and  ${\cal H}^2\leq  {\cal H}^1$.
For instance, we may define the two parameterized sets of graphs ${\cal W}=\{W_{i}\mid i\in{\Bbb{N}}\}$ and ${\cal G}=\{Γ_{i}\mid i\in{\Bbb{N}}\}$
and observe that  ${\cal W}$ and ${\cal G}$ are $\leq_{\sf mn}$-equivalent. However, they are not ${\cal W}$ are $\leq_{\sf tp}$-equivalent
as  ${\cal W}\leq_{\sf tp}{\cal G}$, while ${\cal G}\not\leq_{\sf tp}{\cal W}$.

A set $\mathfrak{A}$ of  $\leq$-{monotone} parameterized sets of graphs is a \emph{$\leq$-antichain} if for every two distinct $\mathcal{H}^1\in\mathfrak{A}$ and $\mathcal{H}^2\in\mathfrak{A}$, neither   ${\cal H}^1\leq {\cal H}^2$  nor  ${\cal H}^2\leq {\cal H}^1$.
Given that $\frak{A}=\{{\cal H}^1,\ldots,{\cal H}^{r}\}$, we define the {\em $i$-th layer} of  $\frak{A}$, denoted by $\frak{A}_{i}$  as the set consisting of the $i$-th graph in each of the  parameterized set of graphs in $\frak{A}$, that is $\frak{A}_{i}=\{H_{i}^{1},\ldots,H_{r}^{i}\}$. 
We also define the graph parameter $\p_{\frak{A},\leq}$ so that $$\p_{\frak{A}}(G)=\max\{i\mid \frak{A}_{i}\leq G\}.$$

\begin{definition}
Let $\leq$ be a graph inclusion relation, $\mathfrak{A}$ be a $\leq$-antichain, and $\mathbf{p}$ be a $\leq$-closed parameter. We say that $\mathfrak{A}$ is a \emph{universal $\leq$-obstruction set} of $\mathbf{p}$ if $\mathbf{p}\sim \mathbf{p}_{\leq,\mathfrak{A}}$.
\end{definition}

For instance $\tw\sim\p_{\leq_{{\sf mn}},\{\cal W\}}\sim\p_{\leq_{\sf mn},\{\cal G\}}$, while $\tw\sim \p_{\leq_{\sf tp},\{\cal W\}}$.
Such type of results exist for all the parameters mentioned in the introduction.
Let ${\cal B}=\{B_{i}, i\in{\Bbb{N}}\}$ be the set of all complete binary trees of height $i$
and let ${\cal S}=\{S_{i}, i\in{\Bbb{N}}\}$  where $S_{i}$ is the {\em $i$-star} graph, that is the graph $K_{1,i}$.
It is also known that $\pw\sim\p_{\leq_{\sf mn},\{{\cal B}\}}$  \cite{RobertsonS83GMI} and that $\cw\sim\p_{\leq_{\sf mn},\{{\cal B},{\cal S}\}}$ \cite{KorachS93tree}. 
Also, according to the result of Wollan in~\cite{Wollan15thest}, $\tcw\sim\p_{\leq_{\sf im},\{{\cal W}\}}$.
Also for a universal obstruction for the parameter of {\sl tree partition width},  see~\cite{DingO96ontre}.\smallskip

In what follows, we  give a family of five parameterized graphs that can serve as a universal $\wtm$-obstruction for edge-treewidth (see \autoref{fig_universal_obstructions}). Given a graph $G=(V,E)$, the \emph{weak {subdivision}} of $G$, denoted $\dot{G}$, is the graph obtained from $G$ by {subdividing} once every edge $xy\in E$ such that both $x$ and $y$ have edge-degree at least $3$ (the {\em subdivision} of an edge is its replacement by a path of length two on the same endpoints). The $i$-fan is the graph composed by a path of length $i$ and a universal vertex.

\begin{enumerate}
\item $\mathcal{H}^1=\{\theta_{3+i}\mid i\in\mathbb{N}\}$, where $\theta_i$ is the graph on two vertices connected by $i$ multiple edges;
\item $\mathcal{H}^2=\{\dot\theta_{3+i}\mid i\in\mathbb{N}\}$;
\item $\mathcal{H}^3=\{\tilde\varphi_{3+i}\mid i\in\mathbb{N}\}$, where $\tilde\varphi_i$ is obtained from the $i$-fan by subdividing once every edge incident to two vertices of vertex-degree $3$;
\item $\mathcal{H}^4=\{\dot\varphi_{3+i}\mid i\in\mathbb{N}\}$, where $\dot\varphi_i$ is the weak subdivision of the $i$-fan;
\item $\mathcal{H}^5=\{\dot\W_{(3+i)\times(3+i)}\mid i\in\mathbb{N}\}$.
\end{enumerate}

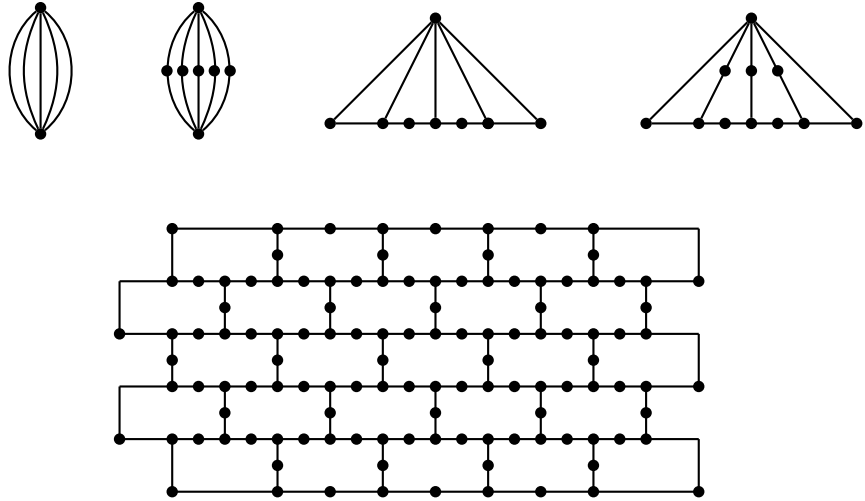
\begin{figure}[h]
\begin{center}
\begin{tikzpicture}[thick,scale=0.7]
\tikzstyle{sommet}=[circle, draw, fill=black, inner sep=0pt, minimum width=3.5pt]

\begin{scope}[xshift=-7.5cm]
\foreach \x/\y in {90/1,270/2}{
\draw node[sommet] (\y) at (\x:1.2){};
}

\draw[-,>=latex] (1) to[bend right=25] (2);
\draw[-,>=latex] (1) to[bend right=50] (2);
\draw[-,>=latex] (1) to (2);
\draw[-,>=latex] (1) to[bend left=25] (2);
\draw[-,>=latex] (1) to[bend left=50] (2);

\end{scope}

\begin{scope}[xshift=-4.5cm]
\foreach \x/\y in {90/1,270/2}{
\draw node[sommet] (\y) at (\x:1.2){};
}
\foreach \x/\y in {-0.6/4,-0.3/5,0/6,0.3/7,0.6/8}{
\draw node[sommet] (\y) at (\x,0){};
}

\draw[-,>=latex] (1) to[bend right=25] (2);
\draw[-,>=latex] (1) to[bend right=50] (2);
\draw[-,>=latex] (1) to (2);
\draw[-,>=latex] (1) to[bend left=25] (2);
\draw[-,>=latex] (1) to[bend left=50] (2);
\end{scope}

%
%
%
%
%
%

\begin{scope}[xshift=0cm,yshift=-1cm]
\foreach \x/\y in {-2/1,-1/2,0/3,1/4,2/5}{
\draw node[sommet] (\y) at (\x,0){};
}

\foreach \x/\y in {-0.5/12,0.5/13,}{
\draw node[sommet] (\y) at (\x,0){};
}

\draw node[sommet] (6) at (0,2){};

\foreach \z in {1,2,3,4,5}{
\draw[-,>=latex] (6) to (\z);
}
\draw[-,>=latex] (1) to (5);
\end{scope}

\begin{scope}[xshift=6cm,yshift=-1cm]
\foreach \x/\y in {-2/1,-1/2,0/3,1/4,2/5}{
\draw node[sommet] (\y) at (\x,0){};
}

\foreach \x/\y in {-0.5/7,0/8,0.5/9}{
\draw node[sommet] (\y) at (\x,1){};
}

\foreach \x/\y in {-0.5/12,0.5/13}{
\draw node[sommet] (\y) at (\x,0){};
}

\draw node[sommet] (16) at (0,2){};

\foreach \z in {1,2,3,4,5}{
\draw[-,>=latex] (16) to (\z);
}
\draw[-,>=latex] (1) to (5);

\end{scope}

\begin{scope}[xshift=0cm,yshift=-8cm]

\foreach \x in {-5,-3,-2,-1,0,1,2,3,5}{
\draw node[sommet] () at (\x,0){};
}
\foreach \x in {-5,-3,-2,-1,0,1,2,3}{
\draw node[sommet] () at (\x,5){};
}

\foreach \x in {-3,-1,1,3}{
\foreach \y in {0.5,4.5}{
\draw node[sommet] () at (\x,\y){};
}
}

\foreach \x in {-6,-5,-4,-3,-2,-1,0,1,2,3,4}{
\foreach \y in {1,3}
\draw node[sommet] () at (\x,\y){};
}

\foreach \x in {-4.5,-3.5,-2.5,-1.5,-0.5,0.5,1.5,2.5,3.5}{
\foreach \y in {1,2,3,4}
\draw node[sommet] () at (\x,\y){};
}

\foreach \x in {-5,-4,-3,-2,-1,0,1,2,3,4,5}{
\foreach \y in {2,4}
\draw node[sommet] () at (\x,\y){};
}

\foreach \x in {-4,-2,0,2,4}{
\foreach \y in {1.5,3.5}{
\draw node[sommet] () at (\x,\y){};
}
}

\foreach \x in {-5,-3,-1,1,3}{
\draw node[sommet] () at (\x,2.5){};
}

\foreach \y in {0,5}{
\draw (-5,\y) -- (5,\y);
}

\foreach \y in {1,2,3,4}{
\draw (-6,\y) -- (5,\y);
}

\foreach \x in {-5,-3,-1,1,3,5}{
\foreach \y in {0,2,4}{
\draw (\x,\y) -- (\x,\y+1);
}
}

\foreach \x in {-6,-4,-2,0,2,4}{
\foreach \y in {1,3}{
\draw (\x,\y) -- (\x,\y+1);
}
}
\end{scope}

\end{tikzpicture}
\end{center}
    \caption{From left to right in the top row, we have the graph $\theta_5$ and its weak subdivision $\dot\theta_5$, the graph $\tilde\varphi_5$ and the graph $\dot\varphi_5$. The weak subdivision of the $6\times 6$-wall is drawn at the bottom row.}
    \label{fig_universal_obstructions}
\end{figure}

\begin{theorem}
\label{paplriok}
The set $\mathfrak{A}=\{\mathcal{H}^j\mid 1\le j\le 5\}$ is a universal $\wtm$-obstruction set for edge-treewidth, that is $\etw\sim \mathbf{p}_{\wtm,\mathfrak{A}}$.
\end{theorem}
\begin{proof}
We first prove that every $\mathcal{H}^j$ is a $\wtm$-monotone parameterized set of graphs (\autoref{cl_monotone}). Then we show that $\mathfrak{A}$ is a $\wtm$-antichain (\autoref{cl_antichain}). Finally, we establish the parametric equivalence between $\etw$ and $\mathbf{p}_{\mathfrak{A}}$ (\autoref{cl_equiv_f} and \autoref{cl_equiv_b}).

\smallskip
\begin{claim} \label{cl_monotone}
For every  $j$, $1\le j\le 5$, $\mathcal{H}^j$ is a $\wtm$-monotone parameterized set of graphs. 
\end{claim}
\noindent
\textit{Proof of claim.} 
Observe that
$\theta_i$ is obtained from $\theta_{i+1}$ by deleting one edge and that $\dot\theta_i$ is obtained from $\dot\theta_{i+1}$ by deleting one vertex of vertex-degree $2$. We note that every graph $H^j_{i+1}$ ($j\in\{3,4\}$) in one of the two fan families $\mathcal{H}^3$, $\mathcal{H}^4$ has an induced path $P$ of at least $4+i$ vertices. Let $x$ be one of the two extreme vertices of that path. To obtain $H^j_i$ from $H^j_{i+1}$, we first delete $x$. If this generate edges between vertices of degree two in the resulting graph, then we contract these (one or two) edges.
{One can easily check that, in the weak subdivision $\dot{W}_{4+i}$ of the wall, performing vertex deletions and edges contraction between degree-two vertices along the top row and the left-most column, we can obtain  $\dot{W}_{(3+i)\times(3+i)}$.}  \hfill \qedclaim

\smallskip
\begin{claim}  \label{cl_antichain}
$\mathfrak{A}=\{\mathcal{H}^j\mid 1\le j\le 5\}$ is a $\wtm$-antichain. 
\end{claim}
\noindent
\textit{Proof of claim.} 
\begin{enumerate}
\item Consider $j\in [1,4]$. It is well-known that the treewidth of any subdivision of $W_{i\times i}$ is $i$. Observe that for every $i\in\mathbb{N}$, $\tw(H^j_i)=2$. As $G\wtm H$ implies that $\tw(G)\le \tw(H)$, we have that $\mathcal{H}^j\not\wtm \mathcal{H}^5$. Moreover, if  $G\wtm H$, then $\Delta_{\sf e}(G)\le \Delta_{\sf e}(H)$. This implies that $\mathcal{H}^5\not\wtm \mathcal{H}^j$.

\item We observe that a weak topological minor of $\theta_i$ is either the single vertex graph or a theta graph $\theta_{i'}$ for some $i'<i$. It follows that for every $j\in\{2,3,4\}$, $\mathcal{H}^j\not\wtm\mathcal{H}^1$.

\item We observe that every weak topological minor of $\dot\theta_i$ that is biconnected is either the single vertex graph or $\dot\theta_{i'}$ for some $i'<i$. It follows that for every $j\in\{1,3,4\}$, $\mathcal{H}^j\not\wtm\mathcal{H}^2$.

\item We observe that for $j\in\{3,4\}$, every weak topological minor of $H^j_i$ that is biconnected is either the $K_3$, the fan $\varphi_3$ or contains an induced path on at least $5$ vertices. It follows that for every $j'\in\{1,2\}$, $\mathcal{H}^{j'}\not\wtm\mathcal{H}^j$.

\item Finally, we observe that every weak topological minor of $\tilde\varphi_j\in\mathcal{H}^4$ that is biconnected contains a vertex whose neighbors belong to a common path. It follows that $\mathcal{H}^4\not\wtm\mathcal{H}^3$. We also observe that  every biconnected weak topological minor of $\dot\varphi_j$ that is not $K_3$ nor the fan $\varphi_3$, contains a vertex of vertex-degree three whose neighbors all have vertex-degree two. This implies that $\mathcal{H}^3\not\wtm\mathcal{H}^4$. 
\hfill \qedclaim
\end{enumerate}

\smallskip
\begin{claim}  \label{cl_equiv_f}
For every $j$, $1\le j\le 5$ and for every $i\in\mathbb{N}$, $\etw(H_i^j)\ge  \sqrt{i}$.
\end{claim}
\noindent
\textit{Proof of claim.} To see this, observe that for every $j\in[1,5]$, the parameterized set of graphs $\mathcal{H}^j$ contains only biconnected graphs. Moreover for every graph $H^j_i\in\mathcal{H}^j$, we have $\max\{\Delta_{\sf e}(H^j_i),\tw(H^j_i)\}\ge i$. It follows from~\autoref{th_equivalence} that $\etw(H^j_i)\ge\sqrt{i}$.
\hfill \qedclaim

\smallskip
\begin{claim}  \label{cl_equiv_b}
There exists a function $f:\mathbb{N}\rightarrow\mathbb{N}$ such that if for every $j\in[1,5]$, $H^j_{k}\not\wtm G$, then $\etw(G)\leq f(k)$.
\end{claim}
\noindent
\textit{Proof of claim.} 
By the parametric equivalence between $\etw(G)$ and $\mathbf{p}(G)=\max\big\{\{\Delta_{\sf e}(B) \:|\: B \in {\sf bc}(G)\} \cup \{ {\tw}(B) \:|\: B \in {\sf bc}(G)\}\big\}$ (see \autoref{th_equivalence}), it suffices to prove that there exists a function $f:\mathbb{N}\rightarrow\mathbb{N}$ such that if for every $j\in[1,5]$, $H^j_{k}\not\wtm G$, then $\mathbf{p}(G)\leq f(k)$. Let $B$  be the block of $G$ such that $\mathbf{p}(G)=\max\{\Delta_{\sf e}(B),\tw(B)\}$. We have two cased to consider:
\begin{itemize}
\item Suppose first that $\mathbf{p}(G)=\tw(B)$. Observe that by the wall theorem (see~\cite{RobertsonS86GMV,ChekuryC16polyn}), if $B$ contains a subdivision of $W_{k\times k}$ as a subgraph, then $\tw(G)\ge k^{\O(1)}$.

\item So let us consider the case where $\mathbf{p}(G)=\Delta_{\sf e}(B)$. Let $x$ be a vertex of $B$ such that $\edeg_{B}(x)=\Delta_{\sf e}(B)$ and suppose that $\edeg_{G}(x)\geq k\cdot k^{2k-3}$. As by assumption, $\theta_k\not\wtm G$, we have that $\vdeg_{B}(x)\ge k^{2k-3}$. Let $T'$ be a (simple) subtree of $B\setminus\{x\}$ spanning all the neighbors of $x$ in $B$. We consider the tree $T$ obtained from $T'$ by {disolving} every vertex of vertex-degree two that is not a neighbor of $x$ in $G$. It follows that every vertex of $T$ is a neighbor of $x$ or a branching vertex. Observe that every vertex $y$ in $T$ has vertex-degree $\vdeg_{T}(y)< k$, as otherwise we would have $\dot\theta_k\wtm G$. It follows that $T$ is a tree spanning at least $k^{2k-3}$ vertices and its maximum vertex-degree $\Delta_{\sf v}(T)$ is at most $k$. Moreover, the number of vertices in $T$ is at most $(\Delta_{\sf v}(T))^{{\lceil\sf diam}(T)/2\rceil}$, where ${\sf diam}(T)$ is the diameter of $T$. This implies that $T$ contains a path $P$ of length at least $4k-6$. Observe that the vertex set $V(P)$ of $P$ is partitioned into $P_x=V(P)\cap N(x)$ and $P_{\overline x}=V(P)\cap\overline{N}(x)$.  Suppose that $|P_x|\geq |P_{\overline{x}}|$. Let $Y$ be the subset of vertices of $P$ containing every odd neighbor of $x$ in $P$. We have that $|Y|\geq k$. Let $H$ be the graph with vertex set $V(P)\cup\{x\}$ and containing the edges of $P$ and the edges between $x$ and the vertices of $Y$. We observe that, by construction of $T$ and definition of $Y$, $\tilde\varphi_k$ is a weak topological minor of $H$ and of $B$: contradiction. 
So let us assume that $|P_x|< |P_{\overline x}|$. Observe that, by construction of $T$, every vertex $y\in V(P)\cap\overline{N}(x)$ is a branching vertex. Thereby for every such vertex $y$, $T$ contains a path $P_y$ of length at least $2$ between $x$ and $y$. Let $Y$ be the subset of vertices of $P$ containing every odd vertex of $V(P)\cap\overline{N}(x)$. We have that $|Y|\geq k$. 
Let $H$ be the graph containing the vertex $x$, the path $P$ and every path $P_y$ between $x$ and $y\in Y$. Again, we observe that, by construction of $T$ and definition of $Y$, $\dot\varphi_k$ is a weak topological minor of $H$ and of $B$: contradiction.
\hfill \qedclaim
\end{itemize}

As $\etw$ is closed under weak topological minor (see~\autoref{th_weak_topological_minor}), by \autoref{cl_equiv_f}, we have that $\mathbf{p}_{\mathfrak{A}}(G)\le \etw(G)^2$. By \autoref{cl_equiv_b}, we have that $\etw(G)\leq \mathbf{p}_{\wtm,\mathfrak{A}}(G)^{4\cdot\mathbf{p}_{\wtm,\mathfrak{A}}(G)+1}$. It follows that $\etw\sim\mathbf{p}_{\wtm,\mathfrak{A}}$.
\end{proof}

Let us discuss the exponential gap obtained above between $\etw$ and $\mathbf{p}_{\mathfrak{A}}$. Consider the graph family $\mathcal{G}=\{G_{3+i}\mid k\in\mathbb{N}\}$ where $G_{3+i}$ is defined as follows (see \autoref{fig_tight_obstruction}). Let $T$ be a tree of diameter $2i+1$ such that every internal node has vertex-degree $i+2$. Then $G_{3+i}$ is obtained by adding a universal vertex $x$ to the leaves of $T$, and then by replacing every edge with $i+2$ multiple edges. 

\begin{figure}[h]
\begin{center}
\begin{tikzpicture}[thick,scale=0.7]
\tikzstyle{sommet}=[circle, draw, fill=black, inner sep=0pt, minimum width=3.5pt]

\begin{scope}[xshift=-3.4cm,yshift=0cm]
\draw node[sommet] (20) at (1.6,2){};
\foreach \x\y in {0.4/10,1.6/11,2.8/12}{
\draw node[sommet] (\y) at (\x,1){};
\draw[-] (20) to (\y);
}
\foreach \x/\y in {0/1,0.4/2,0.8/3,1.2/4,1.6/5,2/6,2.4/7,2.8/8,3.2/9}{
\draw node[sommet] (\y) at (\x,0){};
}
\foreach \y in {0,0.4,0.8}{
\draw[-] (10) to (\y,0) ;
}
\foreach \y in {1.2,1.6,2}{
\draw[-] (11) to (\y,0) ;
}
\foreach \y in {2.4,2.8,3.2}{
\draw[-] (12) to (\y,0) ;
}
\end{scope}

\begin{scope}[xshift=0.2cm,yshift=0cm]
\draw node[sommet] (20) at (1.6,2){};
\foreach \x\y in {0.4/10,1.6/11,2.8/12}{
\draw node[sommet] (\y) at (\x,1){};
\draw[-] (20) to (\y);
}
\foreach \x/\y in {0/1,0.4/2,0.8/3,1.2/4,1.6/5,2/6,2.4/7,2.8/8,3.2/9}{
\draw node[sommet] (\y) at (\x,0){};
}
\foreach \y in {0,0.4,0.8}{
\draw[-] (10) to (\y,0) ;
}
\foreach \y in {1.2,1.6,2}{
\draw[-] (11) to (\y,0) ;
}
\foreach \y in {2.4,2.8,3.2}{
\draw[-] (12) to (\y,0) ;
}
\end{scope}

\draw[-] (20) to (-1.8,2) ;
\draw[-] (20) to (1.8,2) ;

\draw node[sommet] (30) at (0,-2){};
\foreach \x in {-3.4,-3,...,3.4}{
\draw (30) to (\x,0);
}

\node[below] (a) at (0,-2) {$x$};

\end{tikzpicture}
\end{center}
    \caption{The graph $G_5$ is obtained from the above graph by replacing every edge with $4$ multiple edges.}
    \label{fig_tight_obstruction}
\end{figure}
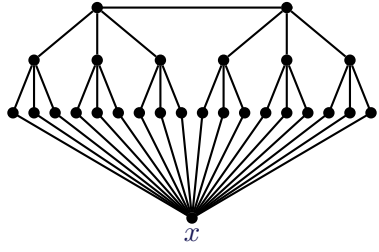

\begin{observation} \label{obs_tight_bound}
For every $k\in\mathbb{N}$, the graph $G_{3+i}$ does not contain as weak topological minor any of $\theta_{3+i}$, $\dot\theta_{3+i}$, $\tilde{\varphi}_{3+i}$, $\dot\varphi_{3+i}$ and $\dot\W_{(3+i)\times(3+i)}$.
\end{observation}
\begin{proof}
As every edge has multiplicity $i+2$, $\theta_{3+i}\not\wtm G_{3+i}$. Since the vertex-degree of every vertex, but $x$, is at most $i+2$, we also have that $\dot\theta_{3+i}\not\wtm G_{3+i}$. The diameter of $G_{3+i}-x$ is $2i+1$ while the diameter of $\tilde{\varphi}_{3+i}$ and of $\dot\varphi_{3+i}$ are $2i+2$. It follows that that  $\tilde{\varphi}_{3+i}\not\wtm G_{3+i}$ and  $\tilde{\varphi}_{3+i}\not\wtm G_{3+i}$. Finally, observe that $G_{3+i}$ is a series-parallel graph (it has treewidth $2$). As a consequence, $\dot\W_{(3+i)\times(3+i)}\not\wtm G_{3+i}$.
\end{proof}

We observe that $\Delta_{\sf e}(G_{3+i})=2\times (i+1)^{i+1}$. On one hand, as $G_{3+i}$ is biconnected and $\tw(G_{3+i})=2$, by \autoref{th_equivalence}, we have that $\etw(G_{3+i})\ge\sqrt{\Delta_{\sf e}(G_{3+i})}$. On the other hand, by \autoref{obs_tight_bound}, we have that $ \mathbf{p}_{\mathfrak{A}}(G_{3+i})=3+i$. This proves that  the exponential gap  between $\etw$ and $\mathbf{p}_{\mathfrak{A}}$ is tight.
\medskip

\paragraph{Incomparability between edge-treewidth and tree-cut width.}
We now prove that the parameters of edge-treewidth and tree-cut width are parametrically incomparable, i.e., $\tcw\not\sim \etw$.
For this, we may avoid giving the (lengthy) definition of tree-cut width and use instead the main result of Wollan in~\cite{Wollan15thest} who proved that $\tcw\sim\p_{\leq_{\sf im},{\{\cal W}\}}$. Therefore it  remains to show why $\p_{\leq_{\sf im},{\{\cal W}\}}\not\sim\etw$.
Notice that $\p_{\leq_{\sf im},{\{\cal W}\}}(θ_{i})$ is trivially bounded for every $i$, while the graphs in $\mathcal{H}^1=\{\theta_{3+i}\mid i\in\mathbb{N}\}$ have unbounded edge-treewidth, because of \autoref{paplriok}. This implies that $\etw\not\lesssim\p_{\leq_{\sf im},{\{\cal W}\}}$. 
Let $Z_{i}$ be the  graph  obtained after taking 
$|E(W_{i})|$ copies of $θ_{3}$ and then identifying one vertex of each copy to a single vertex.
It is  easy to see that $W_{i}\leq_{\sf im}Z_{i}$, therefore  $\p_{\leq_{\sf im},{\{\cal W}\}}\geq i$.
On the other side, each of the $|E(W_{i})|$ blocks of $Z_{i}$ has maximum edge-degree three, and therefore $Z_{i}$ has bounded edge-treewidth, because of \autoref{th_equivalence}, which implies that $\p_{\leq_{\sf im},{\{\cal W}\}}\not\lesssim\etw$.

\section{On the complexity of computing edge-treewidth}
\label{npcomoplo}

We discuss the computational complexity of the following decision problem:

\pbdefn{Edge-Treewidth}{A graph $G$ and an integer $k$.}{Decide if $\etw(G)\le k$?}

\begin{theorem}[${\sf NP}$-completeness]
The problem {\sc Edge-treewidth} is ${\sf NP}$-complete.
\end{theorem}

\begin{proof}
Clearly, the problem {\sc Edge-treewidth} belongs to {\sf NP}. Indeed given a layout ${\sf L}$ of $G$, one can check in polynomial time if $\mathbf{p}_{\delta^{\sf ce}}(G,{\sf L})\le k$.
The  ${\sf NP}$-hardness proof is based on a reduction from the {\sc Min-Bisection} problem.

\pbdefn{Min-bisection problem}{A graph $G$ and an integer $k$.}{Is there a bipartition $V_1,V_2$ of $V(G)$ such that $||V_1|-|V_2||\le 1$ and $|E_{G}(V_1)| \le k$
?
}

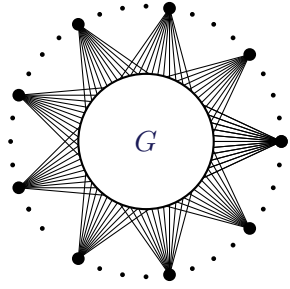
\begin{figure}[h]
\centering
\begin{center}
\begin{tikzpicture}[thick,scale=0.6]
\tikzstyle{sommet}=[circle, draw, fill=black, inner sep=0pt, minimum width=4pt]
\tikzstyle{point}=[circle, draw, fill=black, inner sep=0pt, minimum width=1pt]

\foreach \x in {0,40,...,360}{
\begin{scope}[shift=(\x:3)]
\foreach \y in {-25,-20,...,25}{
\draw[-,thin] (0:0) to (180+\x+\y:3){};
}
\end{scope}
}

\draw[fill=white] (0,0) circle (1.5) ;
\node[] (a) at (0,0) {$G$};

\foreach \x in {0,40,...,360}{
\draw node[sommet] () at (\x:3) {};
}
\foreach \x in {10,20,30,50,60,70,90,100,110,130,140,150,170,180,190,210,210,220,250,260,270,290,300,310,330,340,350}{
\draw node[point] () at (\x:3) {};
}

\end{tikzpicture}
\end{center}
\caption{The graph $H$ is obtained by adding to the graph $G$ an independent set of size $n^2$ and making each of the vertices of the independent set universal to the vertices of $G$.
}
\label{reductionex}
\end{figure}

Let $(G,k)$ be an instance of {\sc Min Bisection} where $G$ is a graph on $n$ vertices. We build an instance $(H,w)$ of {\sc Edge-treewidth} as follows (see \autoref{reductionex}):
\begin{itemize}
\item $V(H)=V(G)\cup Q$, where $Q$ is an independent set of size $n^2$;
\item $E(H)=E(G)\cup \{\{x,y\}\mid x\in V(G), y\in Q\}$, that is we add to $E(G)$ all possible edges between $Q$ and $V(G)$ ;
\item and $w= \frac{n^3}{2} + k$.
\end{itemize}

\begin{claim} \label{cl_NP_f}
If $(G,k)$ is a \textsf{YES}-instance of {\sc Min-Bisection}, then $(H,w)$ is a \textsf{YES}-instance of {\sc Edge-treewidth}.
\end{claim}
\noindent
\textit{Proof of claim.} 
Let $V_1,V_2$ be a bipartition of $V(G)$ such that $||V_1|-|V_2||\le 1$ and $|E_{G}(V_1)| \le k$. Let us consider a layout ${\sf L}$ of $H$ such that for every triple $x,y,z\in V(H)$ such that $x\in V_1$, $y\in Q$, $z\in V_2$, then $x\prec_{{\sf L}} y\prec_{{\sf L}} z$, {(in other words ${\sf L}=\langle V_1\cdot Q\cdot V_2\rangle$)} (here we use $\cdot$ for the sequence concatenation operation). {Observe that ${δ^{\sf ec}_{H,{\sf L}}(\frac{n}{2}+1)}\le\frac{n^3}{2}+k=w$.}
Indeed, $S_{\frac{n}{2}+1}=Q\cup V_2$ induces a connected subgraph of $G$, that is $C_H(S_{\frac{n}{2}+1},x_{\frac{n}{2}+1})=Q\cup V_2$. Moreover, there are $n^2\cdot\frac{n}{2}$ edges between $V_1$ and $Q$ and at most $k$ edges between $V_1$ and $V_2$. Let us prove that for every $i\neq \frac{n}{2}+1$, $δ^{\sf ec}(H,{\sf L},i)\le w$. We consider three cases:
\begin{itemize}
\item Suppose that $i \le \frac{n}{2}$ (i.e. $x_i\in V_1$).  Observe that $C_H(S_i,x_i)=S_i$ and thereby $E_H(C_H(S_i,x_i))$ contains $(i-1)\cdot n^2$ between $V(H)\setminus S_i$ and $Q$, and $c$ edges between $V(H)\setminus S_i$ and $S_i\setminus Q$ with $c\le \frac{n^2}{2}$. All in all, we have that $δ^{\sf ec}(H,{\sf L},i)\le (i-1)\cdot n^2 + \frac{n^2}{2} \le (\frac{n}{2} - 1)\cdot n^2 + \frac{n^2}{2} = \frac{n^3}{2} - \frac{n^2}{2}$.         
\item Suppose that $i > n^2 + \frac{n}{2}$ (i.e. $x_i\in V_2$). A symmetric argument proves that $δ^{\sf ec}(H,{\sf L},i)\le w$.

\item Suppose that $\frac{n}{2}+1<i\le n^2+\frac{n}{2}$ (i.e. $x_i\in Q$). Then $E_H(C_H(S_i,x_i))$ contains at most 
$k$ edges between the vertices of $V_1$ and $V_2$, $(i-\frac{n}{2})\cdot \frac{n}{2}$ edges  between the vertices of $Q$ and $V_2$, and $(n^2+\frac{n}{2} - i)\cdot \frac{n}{2}$ edges between the vertices of $V_1$ and $Q$. In total, we obtain that $δ^{\sf ec}(H,{\sf L},i)\le w$.
\end{itemize}
So we proved that $δ^{\sf ec}(H,{\sf L})\le w$ and thereby $(H,w)$ is a  \textsf{YES}-instance of {\sc Edge-treewidth}. \hfill\qedclaim

\smallskip
\begin{claim} \label{cl_NP_b}
If  $(H,w)$ is a \textsf{YES}-instance of {\sc Edge-treewidth}, then $(G,k)$ is a \textsf{YES}-instance of {\sc Min-Bisection}.
\end{claim}
\noindent
\textit{Proof of claim.} 
Let $L$ be a layout of $H$ such that $δ^{\sf ec}(H,{\sf L})\le w=\frac{n^3}{2} + k$ with $k \le \frac{n^2}{2}$. Let $i$ be the smallest index for which $|(V(H)\setminus S_i)\cap V(G)| = |S_i\cap V(G)|=\frac{n}{2}$. We denote $V_1 = (V(H)\setminus S_i)\cap V(G)$ and $V_2=S_i\cap V(G)$. Let us prove that the bipartition $V_1,V_2$ certifies that $(G,k)$ is a \textsf{YES}-instance of {\sc Min-Bisection}.

We first argue that $S_i\cap Q \ne \emptyset$. For the sake of contradiction, suppose that $S_i\cap Q = \emptyset$. 
Let $j<i$ be the largest integer such that $x_j\in Q$. By the minimality of $i$, we know that $v_{i-1}\notin Q$ and thereby $j<i-1$. We observe that $C_H(S_{j-1},x_{j-1}))=S_{j-1}$ and moreover $Q\setminus \{v_j\} \subseteq V(H)\setminus S_j$ and $|S_{i-1}\cap V(G)|\ge \frac{n}{2}+1$. It follows that $E_H(S_{j-1})$ contains at least $(n^2-1)\cdot (\frac{n}{2}+1)$ edges, that is $δ^{\sf ec}(H,{\sf L},i)>\frac{n^3}{2} + \frac{n^2}{2}$, contradicting the assumption that $δ^{\sf ec}(H,{\sf L})\le w=\frac{n^3}{2} + k$ with $k \le \frac{n^2}{2}$.

As a consequence of $Q\cap S_i\neq\emptyset$, we have that $C_H(S_i,x_i)=S_i$. This implies that at least half of the edges incident to a vertex $q$ of $Q$ belongs $E_H(C_H(S_i,x_i))$. Suppose that $q\in S_i$, then every edge between $q$ and $V_1$ belongs to $E_H(C_H(S_i,x_i))$. Suppose that $q\notin S_i$, then every edge between $q$ and $V_2$ belongs to $E_H(C_H(S_i,x_i))$. It follows that $δ^{\sf ec}(H,{\sf L},i)\ge n^2\cdot \frac{n}{2}=\frac{n^3}{2}$. As by assumption $δ^{\sf ec}(H,{\sf L},i)\le \frac{n^3}{2} + k$, this implies that there are at most
$k$ edges betwwen $V_1$ and $V_2$, certifying that $(G,k)$ is a \textsf{YES}-instance of {\sc Min-Bisection}.
\hfill\qedclaim

\smallskip
 \autoref{cl_NP_f} and \autoref{cl_NP_b} establish the \textsf{NP}-hardness of the {\sc Edge-treewidth} problem.
\end{proof}

%
%
%
%

%
%
%
%


\end{document}